\documentclass[a4paper,reqno,12pt]{amsart} 

\usepackage[english]{babel}
\usepackage[utf8]{inputenc}
\usepackage[T1]{fontenc}
\usepackage{eucal}
\usepackage{xparse}
\usepackage{microtype}
\usepackage{csquotes}
\usepackage{verbatim}
\usepackage{amsthm}
\usepackage{mathtools}
\usepackage{amssymb,amsfonts}
\usepackage{tikz}
\usetikzlibrary{cd}
\usepackage[backend=biber,giveninits=true,maxnames=5,doi=true,isbn=false,url=false,sortcites,style=phys,biblabel=brackets,sorting=nyt]{biblatex} 
\AtEveryBibitem{%
    \clearlist{language}
}

\usepackage[breaklinks,unicode=true]{hyperref}
\usepackage[capitalise]{cleveref} % \cref{…} (References)\usepackage{cleveref}

\DeclarePairedDelimiter{\parens}{\lparen}{\rparen}
\DeclarePairedDelimiter{\bracks}{\langle}{\rangle}
\DeclarePairedDelimiter{\set}{\{}{\}}
\DeclarePairedDelimiter{\sqbracks}{[}{]}

\newcommand*{\toinj}{\hookrightarrow}

\DeclareMathOperator{\ev}{ev} %Evaluation map

\newcommand*{\ZZ}{\mathbb{Z}}

\newcommand*{\RR}{\mathbb{R}}

\DeclareMathOperator{\Id}{Id}
\newcommand*{\comp}{\circ} % Function composition
\newcommand*{\restr}[2]{#1|_{#2}}

\DeclareMathOperator{\supp}{supp} %Support

\DeclareMathOperator{\Cont}{\mathcal{C}}

\newcommand{\dd}{\mathrm{d}}
\NewDocumentCommand{\dv}{mm}{\frac{\dd #1}{\dd #2}}
\NewDocumentCommand{\pdv}{mm}{\frac{\partial #1}{\partial #2}}

\NewDocumentCommand{\orth}{om}{{#2}^{\bot\IfValueT{#1}{_#1}}}
\NewDocumentCommand{\orthL}{om}{\prescript{\bot\IfValueT{#1}{_#1}}{}{#2}}

\DeclarePairedDelimiter{\gen}{\langle}{\rangle}

\newcommand*{\lieBr}[1]{\sqbracks{#1}}
\newcommand*{\lieD}[1]{\mathcal{L}_{#1}}

\DeclareMathOperator{\im}{Im}
\newcommand*{\ann}[1]{{#1}^{\circ}}

\newcommand{\contr}[1]{\iota_{#1}}
\newcommand*{\Forms}{\Omega}
\newcommand*{\VecFields}{\mathfrak{X}}

\theoremstyle{plain}
\newtheorem{theorem}{Theorem}
\newtheorem*{theorem*}{Theorem}
\newtheorem{lemma}[theorem]{Lemma}
\newtheorem*{lemma*}{Lemma}
\newtheorem{proposition}[theorem]{Proposition}
\newtheorem*{proposition*}{Proposition}

\newtheorem*{cor*}{Corollary}

\theoremstyle{definition}
\newtheorem{definition}{Definition}
\newtheorem*{definition*}{Definition}

\newtheorem*{example*}{Example}

\crefname{theorem}{Theorem}{Theorems}
\Crefname{theorem}{Theorem}{Theorems}
\crefname{lemma}{Lemma}{Lemmas}
\Crefname{lemma}{Lemma}{Lemmas}
\crefname{proposition}{Proposition}{Propositions}
\Crefname{Prop}{Proposition}{Propositions}
\crefname{cor}{Corollary}{Corollaries}
\Crefname{cor}{Corollary}{Corollaries}
\crefname{definition}{Definition}{Definitions}
\Crefname{definition}{Definition}{Definitions}
\crefname{example}{Example}{Examples}
\Crefname{example}{Example}{Examples}
\crefname{theorem}{Theorem}{Theorems}

\newtheorem*{exercise*}{Exercise}
\crefname{exercise}{exercise}{exercises}
\Crefname{exercise}{Exercise}{Exercises}  

\theoremstyle{remark}
\newtheorem{remarkx}{Remark}
\newtheorem*{remarkx*}{Remark}
\crefname{remark}{Remark}{Remarks}
\Crefname{remark}{Remark}{Remarks}
\newenvironment{remark}
  {\pushQED{\qed}\remarkx}
  {\popQED\endremarkx}
\newenvironment{remark*}
  {\pushQED{\qed}\remarkx*}
  {\popQED\endremarkx*}

    \newcommand{\CharD}{\mathcal{C}}
    
    \newcommand{\Reeb}{\mathcal{R}}
    \newcommand{\sects}{\Gamma}
    \newcommand{\Fib}{\mathrm{F}}
    \newcommand{\lsharp}{\sharp_\Lambda}
    \DeclareMathOperator{\Der}{Der}
    \DeclareMathOperator{\grad}{grad}
    \DeclarePairedDelimiter{\jacBr}{\lbrace}{\rbrace}
    \DeclarePairedDelimiter{\dirBr}{\lbrace}{\rbrace_{DJ}}

    \title{Singular Lagrangians and precontact Hamiltonian systems}

   \author[M. de León]{Manuel de León}
   \address{Manuel de Le\'on: Instituto de Ciencias Matem\'aticas (CSIC-UAM-UC3M-UCM),
   c\textbackslash Nicol\'as Cabrera, 13-15, Campus Cantoblanco, UAM
   28049 Madrid, Spain \newline
   and \newline
   Real Academia de Ciencias Exactas, Físicas y Naturales, c\textbackslash de Valverde,
   22, 28004 Madrid, Spain
   } \email{mdeleon@icmat.es}

   \author[M. Lainz Valcázar]{Manuel Lainz Valcázar}
   \address{Manuel Laínz:
   Instituto de Ciencias Matem\'aticas (CSIC-UAM-UC3M-UCM),
   c$\backslash$ Nicol\'as Cabrera, 13-15, Campus Cantoblanco, UAM
   28049 Madrid, Spain} \email{manuel.lainz@icmat.es}

    \date{\today} % Today's date or a custom date
\begin{document}

\begin{abstract}
    In this paper we discuss singular Lagrangian systems on the framework of contact geometry. These systems exhibit a dissipative behavior in contrast with the symplectic scenario. We develop a constraint algorithm similar to the presymplectic one studied by Gotay and Nester (the geometrization of the well-known Dirac-Bergman algorithm). We also construct the Hamiltonian counterpart and prove the equivalence with the Lagrangian side. A Dirac-Jacobi bracket is constructed similar to the Dirac bracket.
\end{abstract}
\maketitle

{\small\tableofcontents}
\section{Introduction}
As it is well-known, the description of Hamiltonian mechanics is developed on symplectic manifolds; indeed, given a symplectic manifold $(M, \omega)$ and a Hamiltonian function $H:M\to\RR$, we obtain the dynamics as the Hamiltonian vector field given by the equation $\contr{X_H}  \omega = dH$. The integral curves of $X_H$ satisfy Hamilton equations. 

This geometric framework has two formulations: if $M$ is the cotangent bundle $T^*Q$ of a configuration manifold $Q$, equipped with its canonical symplectic form $\omega_Q$, we obtain a classical Hamiltonian mechanical system; and, if $M$ is the tangent bundle $TQ$ equipped with the symplectic form $\omega_L$ constructed from a regular Lagrangian $L : TQ \to \RR$, then we obtain the usual Euler-Lagrange equations (in the latter case, the vector field providing the dynamics is the solution $\xi_L$ of the equation $\contr{\xi_L} \, \omega_L = dE_L$, where $E_L$ is the energy of the system; $\xi_L$ is a second order differential equation on $TQ$ whose solutions are the ones of the Euler-Lagrange equations). Both sides are related by the Legendre transformation.

A fascinating scenario occurs when the Lagrangian function is not regular, that is, its Hessian matrix with respect to the velocities is singular. Hence, the 2-form $\omega_L$ is not symplectic, and the equation $\contr{X}  \omega_L = dE_L$ has not solution in general, or even, if there is a solution, it is not unique. In order to deal with singular Lagrangians, and motivated for the need to study the quantization of electromagnetism, P.A.M. Dirac developed a constraint algorithm (now called Dirac-Bergmann algorithm) that allows us to construct the dynamics of the system~\cite{Dirac1950}. This constraint algorithm has been later geometrized by M.J.~Gotay and J.M.~Nester~\cite{Gotay1978}. 

The geometric version of the algorithm relies on the concept of presymplectic systems, that is, a closed 2-form $\omega$ on a manifold $M$ which is not symplectic but has constant rank. So we analyze Hamilton equations
\begin{equation}
	i_X \, \omega = dH
\end{equation}
for a Hamiltonian function $H$ on $M$. We consider the points where there is a solution of the above equation, and so we obtain a constraint submanifold $M_1$ along which there is a solution. But the dynamics should be tangent to $M_1$ so we have to restrict ourselves to those points in $M_1$ where a solution exists but it is tangent to $M_1$. The algorithm continues and, in the favorable cases, it stabilizes at some level, $M_{i+1} = M_i$ which is called the final constraint submanifold.

The above algorithm can be applied to the case of singular Lagrangian systems, but, when de Lagrangian satisfies some weak regularity condition, we can also develop a Hamiltonian counterpart and the corresponding constraint algorithm. Both algorithms are conveniently related by the Legendre transformation.

In addition, Dirac introduced two kind of constraints, first and second class. The second class constraints permits to define a Poisson bracket (called Dirac bracket) that gives the dynamics of the constrained system just as in the classical case with the canonical Poisson bracket.

The goal of the present paper is to extend these constructions to a new kind of Hamiltonian systems, those called contact Hamiltonian systems. Those systems have found applications in many areas, such as irreversible thermodynamics~\cite{Grmela2014}, statistical mechanics~\cite{Bravetti2016}, geometric optics~\cite{Carinena1996} as well as systems with dissipative forces linear in the velocities (Rayleigh dissipation). The Lagrangian formulation has been considered in~\cite{Liu2018,Vermeeren2019}. A recent review of the applications of contact Hamiltonian systems can be found in~\cite{Bravetti2017}.

 Indeed, in the contact case, the underlying geometry is a contact manifold $(M, \eta)$, where $\eta$ is a contact form. This means that $\eta \wedge (d\eta)^n \not= 0$, where $M$ has dimension $2n+1$. So, if we consider a Hamiltonian function $H$ on $M$, the corresponding Hamiltonian vector field $X_H$ solves  the equation
\begin{equation}\label{eq:intro_contact}
    \flat(X) = \dd H - (H + \Reeb(H))\eta,
\end{equation}
where $\Reeb$ is the Reeb vector field and $\flat(X)=  \contr{X} \dd\eta + \eta(X)\eta$.

This equation shows a dissipative behavior which contrasts with the conservative nature of symplectic systems. Therefore, symplectic and contact geometries provide very different dynamics.

We are interested in extending the theory of singular Lagrangian systems to the contact context. As in the usual case, a Lagrangian $L:TQ\times\RR$ is singular when the Hessian matrix with respect to the velocities is not regular. Consequently, the $1$-form $\eta_L$ in $TQ\times\RR$, constructed from de Lagrangian $L$ with the expression
\begin{equation}
  \eta_L = \dd z -  \frac{\partial L}{\partial \dot{q}^i}  \dd q^i,
\end{equation}
is a contact form if and only if $L$ is regular. Therefore, in the singular case, \cref{eq:intro_contact} might have no solution, or the solution might be not unique. In order to study the solutions of \cref{eq:intro_contact}, we develop a constraint algorithm, similar to the Gotay-Nester one. Indeed, we also extend the Dirac-Bergmann algorithm.

The corresponding Hamiltonian picture is developed, assuming some weak regularity conditions on the Lagrangian, so that both approaches are proven to be equivalent.

An interesting result is that, since the ambient bracket of functions on $T^* Q \times \RR$ is Jacobi but not Poisson, we have to introduce the so-called Dirac-Jacobi bracket, which again is Jacobi but not Poisson. This bracket shares some properties with the Dirac bracket and is motivated by the classification of constraints in first and second class.

The paper is structured as follows. In \cref{sec:geometry_and_dynamics} we will introduce the dynamics of Hamiltonian systems on symplectic, cosymplectic and contact manifolds, and explain how they fit on the more general framework of Jacobi manifolds. Furthermore, we will introduce the Jacobi brackets and some of their properties.
In \cref{sec:lagrangian} we present the Hamiltonian formalism for the 
previously introduced geometries.
In~\cref{sec:Hamiltonian_Legendre_transform} we quickly recall the Hamiltonian 
formalism and connect it to the Lagrangian formalism via the Legendre 
transformation.
In~\cref{sec:variational} we prove the equivalence of the Lagrangian formalism 
with Herglotz's variational problem, which shows some interesting connections of contact 
Hamiltonian systems with control theory and calculus of variations.

The rest of the paper deals with singular systems in the precontact geometry setting. In~\cref{sec:precontact_maifolds} we introduce the concept of precontact 
manifold and some of its properties.
In~\cref{sec:constr_alg} we develop a constraint algorithm for precontact 
systems analogous to the Gotay-Nester algorithm for the presymplectic case. We 
also investigate the tangency of the Reeb vector field to the constraint 
submanifolds.
Then, in~\cref{sec:equiv_problem} we see that the Legendre transformation connects the 
Hamiltonian and the Lagrangian sides of the problem in a way that commutes with 
the constraint algorithm, provided some weak regularity conditions on the 
Lagrangian.
In~\cref{sec:DJ_brackets} we introduce the contact version of the Dirac 
brackets, which we call Dirac-Jacobi brackets. This brackets are Jacobi, but, 
contrary to the symplectic case, they are not Poisson. We are also able to 
classify the constraint functions in first or second class, depending on 
weather they carry dynamical information or not.
Then, in~\cref{sec:second_order} we construct explicitly a submanifold $S$ of the final 
constraint manifold, such that there is a unique solution to the 
equations of motion that satisfy the second order differential equation condition along $S$.
Finally, in~\cref{sec:examples} we provide examples with explicit computations of the constraints and the Dirac-Jacobi brackets.

In this article, given a smooth function $F:M\to N$, we denote by $F_*: TM \to TN$ the tangent map. Also, $\contr{X}(\alpha)$ is the contraction of a vector field $X$ and a differential form $\alpha$. We allow distributions $\Delta \subseteq TM$ to be defined along a submanifold and we denote by $\Gamma{(\Delta)}$ de set of vector fields on $M$ tangent to the distribution.

\section{Geometry and dynamics}\label{sec:geometry_and_dynamics}

\subsection{Dynamics on symplectic geometry}
As it is well known, Hamiltonian dynamics is developed using symplectic geometry. Indeed, let $(M, \omega)$ be a symplectic manifold, that is, 
$\omega$ is a non--degenerate closed 2-form, say $d \omega = 0$ and $\omega^n \not= 0$, where $M$ has even dimension $2n$.
Then, if $H : M \to \mathbb R$ is a Hamiltonian function, the Hamiltonian vector field $X_H$
is obtained from the equation
\begin{equation}\label{hamiltonian_vf_symp}
\flat (X_H) = \dd H,
\end{equation}
where $\flat$ is the vector bundle isomorphism
\begin{equation}
    \begin{aligned}
        \flat : TM &\to  T^* M, \\
        v &\mapsto \contr{v}  \omega.
    \end{aligned}
\end{equation}

In Darboux coordinates $(q^i, p_i)$ we have
\begin{align*}
 \omega &= \dd q^i \wedge \dd p_i, \\
 X_H &= \frac{\partial H}{\partial p_i} \frac{\partial}{\partial q^i} - 
\frac{\partial H}{\partial q^i} \frac{\partial}{\partial p_i} 
\end{align*}
in such a way that an integral curve $(q^i(t), p_i(t))$ of $X_H$ satisfies Hamilton equations:
\begin{subequations}
    \begin{align}\label{hsympl2}
    \frac{\dd q^i}{\dd t} & =  \frac{\partial H}{\partial p_i}, \\
    \frac{\dd p_i}{\dd t} & =  - \frac{\partial H}{\partial q^i}.
    \end{align}
\end{subequations}

\subsection{Dynamics on cosymplectic geometry}

A cosymplectic structure on a $(2n+1)$-dimensional manifold $M$ is a pair $(\Omega, \eta)$ where 
$\Omega$ is a closed 2-form, $\eta$ is a closed 1-form, and
$\eta \wedge \Omega^n \not= 0$. 
$(M, \Omega, \eta)$ will be called a cosymplectic manifold.

There is a Darboux theorem for cosymplectic manifolds, that is,
there are local coordinates (called Darboux coordinates) $(q^i, p_i, z)$ in a neighborhood of any point of $M$
such that
\begin{equation}
	\Omega = \dd q^i \wedge \dd p_i, \quad \eta = \dd z.
\end{equation}
There also exist a a unique vector field (called Reeb vector field) $\mathcal R$ such that
\begin{equation}
	i_{\mathcal R} \, \Omega = 0, \quad i_{\mathcal R}\, \eta = 1.
\end{equation}
In Darboux coordinates, we have
\begin{equation}
	\mathcal R = \frac{\partial}{\partial z}.
\end{equation}

Let $H : M \to \mathbb R$ be a Hamiltonian function, say $H = H(q^i, p_i, z)$.
Consider the vector bundle isomorphism

\begin{equation}
    \begin{aligned}
        \tilde{\flat} : TM &\to T^* M,\\
         v &\mapsto \contr{v} \Omega + \eta (v)  \eta
    \end{aligned}
\end{equation}
and define the gradient of $H$ by

\begin{equation}
    \tilde{\flat}(\grad H) = \dd H.
\end{equation}

Then,
\begin{equation}\label{hcosymp}
\grad H = \frac{\partial H}{\partial p_i} \frac{\partial}{\partial q^i} - 
\frac{\partial H}{\partial q^i} \frac{\partial}{\partial p_i} + \frac{\partial H}{\partial z} \, \frac{\partial}{\partial z}.
\end{equation}

Next, we can define two more vector fields:

\begin{itemize}

\item the Hamiltonian vector field

\begin{equation}\label{hamiltonian_vf_cosymp}
    X_H = \grad H - \mathcal R (H) \mathcal R,
\end{equation}

\item and the evolution vector field
\begin{equation}
    {\mathcal E}_H = X_H + {\mathcal R}.
\end{equation}
\end{itemize}

From \cref{hcosymp} we obtain the local expression

\begin{equation}\label{hcosymp2}
{\mathcal E}_H = \frac{\partial H}{\partial p_i} \frac{\partial}{\partial q^i} - 
\frac{\partial H}{\partial q^i} \frac{\partial}{\partial p_i} +  \frac{\partial}{\partial z}.
\end{equation}
Therefore, an integral curve $(q^i(t), p_i(t), z(t))$ of ${\mathcal E}_H$ satisfies the 
time-dependent Hamilton equations:
\begin{subequations}
    \begin{align}\label{hsympl3}
    \frac{\dd q^i}{\dd t} & =  \frac{\partial H}{\partial p_i}, \\
    \frac{\dd p_i}{\dd t} & =  - \frac{\partial H}{\partial q^i},\\
    \frac{\dd z}{\dd t} & =  1,
    \end{align}
\end{subequations}
and then $z=t+\text{constant}$ so that both coordinates can be identified in what concerns with derivatives with respect to $z$ or $t$.

\subsection{Dynamics on contact geometry}

Consider now a contact manifold $(M, \eta)$ with contact form $\eta$; this means that
$\eta \wedge (\dd\eta)^n \not= 0$ and $M$ has odd dimension $2n+1$.
Then, there exist a unique vector field $\mathcal R$ (also called Reeb vector field) such that
\begin{equation}
	i_{\mathcal R} \, \dd \eta = 0 \; , \; i_{\mathcal R}\, \eta = 1.
\end{equation}

There is a Darboux theorem for contact manifolds so that around each point in $M$ one can find local coordinates 
(called Darboux coordinates) $(q^i, p_i, z)$ 
such that
\begin{equation}
	 \eta = \dd z - p_i \, \dd q^i.
\end{equation}

In Darboux coordinates we have
\begin{equation}
	\mathcal R = \frac{\partial}{\partial z}.
\end{equation}

Define now the vector bundle isomorphism
\begin{align*}
    \bar{\flat} : TM &\to T^* M ,\\
     v &\mapsto \contr{v}  \dd \eta + \eta (v)  \eta.
\end{align*}
Notice that $\bar{\flat}(\Reeb)=\eta$.

For a Hamiltonian function $H$ on $M$ we define the Hamiltonian vector field by

\begin{equation}\label{hamiltonian_vf_contact}
    \bar{\flat} (X_H) = \dd H - (\mathcal R (H) + H) \, \eta
\end{equation}.

In Darboux coordinates we get this local expression

\begin{equation}\label{eq:contact_hamiltonian_vf_darboux}
X_H = \frac{\partial H}{\partial p_i} \frac{\partial}{\partial q^i} - 
\parens*{\frac{\partial H}{\partial q^i} + p_i \frac{\partial H}{\partial z}}  \frac{\partial}{\partial p_i} + 
\parens*{p_i \frac{\partial H}{\partial p_i} - H} \frac{\partial}{\partial z}.
\end{equation}
Therefore, an integral curve $(q^i(t), p_i(t), z(t))$ of $X_H$ satisfies the 
dissipative Hamilton equations
\begin{subequations}
    \begin{align*}\label{hcont3}
    \frac{\dd q^i}{\dd t} & =  \frac{\partial H}{\partial p_i}, \\
    \frac{\dd p_i}{\dd t} & =  - \frac{\partial H}{\partial q^i} - p_i \frac{\partial H}{\partial z},\\
    \frac{\dd z}{\dd t} & =  p_i \frac{\partial H}{\partial p_i} - H.
    \end{align*} 
\end{subequations}

The systems of equations presented so far can be understood as instances of Hamiltonian systems on Jacobi manifolds, which we introduce below.

\subsection{Dynamics on Jacobi manifolds}

All the geometric structures and dynamical systems mentioned above are particular examples of a more general kind of geometric structures~\cite{Lichnerowicz1978,deLeon2017,Vaisman2002}, the so-called Jacobi manifolds, whose definition we recall below. 

\begin{definition}\label{def:jacobi_mfd}
    A \emph{Jacobi manifold} is a triple $(M,\Lambda,E)$, where $\Lambda$ is a bivector field (that is, a skew-symmetric contravariant 2-tensor field) and $E \in \VecFields (M)$ is a vector field, so that the following identities are satisfied:
    \begin{align}
        \lieBr{\Lambda,\Lambda} &= 2 E \wedge \Lambda\\
        \lieD{E} \Lambda &= \lieBr{E,\Lambda} = 0,
    \end{align}
    where $\lieBr{\cdot,\cdot}$ is the Schouten–Nijenhuis bracket~\cite{Schouten1953,Nijenhuis1955}.
\end{definition}

Both symplectic, cosymplectic and contact manifolds can be understood as instances of Jacobi manifolds by taking,
\begin{itemize}
    \item For symplectic manifolds:
    \begin{equation}
        \Lambda(\alpha,\beta) = 
        \omega ({\flat}^{-1} (\alpha), {\flat}^{-1}(\beta)), \quad
        E = 0.
    \end{equation}
    \item For cosymplectic manifolds:
    \begin{equation}
        \Lambda(\alpha,\beta) = 
        \Omega ({\tilde\flat}^{-1} (\alpha), {\tilde\flat}^{-1}(\beta)), \quad
        E = 0.
    \end{equation}
    \item For contact manifolds:
    \begin{equation}\label{eq:contact_jacobi}
        \Lambda(\alpha,\beta) = 
        -\dd \eta ({\bar\flat}^{-1} (\alpha), {\bar\flat}^{-1}(\beta)), \quad
        E = - \Reeb.
    \end{equation}    
\end{itemize}

Jacobi manifolds such that $E=0$ are called \emph{Poisson manifolds}, as in the case of symplectic and cosymplectic manifolds. In addition to contact manifolds, another important example of non-Poisson Jacobi manifolds are locally conformal symplectic manifolds~\cite{deLeon2017}.

The Jacobi bivector $\Lambda$ induces a vector bundle morphism between covectors and vectors.
\begin{equation}
    \begin{aligned}
        \lsharp: T M^* &\to     T M\\
        \alpha &\mapsto \Lambda(\alpha, \cdot ).
    \end{aligned}
\end{equation}
This map is an isomorphism in the case of symplectic and cosymplectic manifolds and coincides with $\flat^{-1}$ and $\tilde{\flat}^{-1}$, respectively.
In the case of a contact manifold, $\lsharp$ is not invertible. In fact, $\ker \lsharp = \gen{\eta}$ and $\im \lsharp = \ker{\eta}$. The map $\lsharp$ can be written more directly in terms of the contact structure~\cite[Section~3]{LainzValcazar2018} as:
    \begin{equation}
        \lsharp (\alpha) =
        {\bar{\flat}}^{-1} (\alpha) - \alpha(\Reeb) \Reeb.
    \end{equation}

We can define the Hamiltonian vector field $X_H$ of a Hamiltonian function $H\in \Cont^\infty(M)$ in the context of Jacobi manifolds by taking
\begin{equation}
     X_H = \lsharp(\dd H) + H E.
\end{equation}
By a simple computation one can check that this definition coincides with the usual ones for symplectic, cosymplectic and contact manifolds (\cref{hamiltonian_vf_symp,hamiltonian_vf_cosymp,hamiltonian_vf_contact}).

\subsubsection{Jacobi Brackets}
The Jacobi structure can be characterized in terms of a Lie bracket on the space of functions $\Cont^\infty(M)$, the so-called \emph{Jacobi bracket}.
\begin{definition}\label{def:jac_bra}
    A \emph{Jacobi bracket} $\jacBr{\cdot,\cdot}: \Cont^\infty(M) \times \Cont^{\infty}(M) \to \Cont^\infty(M)$ on a manifold $M$ is a map that satisfies
    \begin{enumerate}
        \item $(\Cont^\infty(M),\jacBr{\cdot,\cdot})$ is a Lie algebra. That is, $\jacBr{\cdot,\cdot}$ is $\RR$-bilinear, antisymmetric and satisfies the Jacobi identity:
        \begin{equation}
            \jacBr{f,g} + \jacBr{g,h} + \jacBr{h,f} = 0
        \end{equation}
        for arbitrary $f,g,h \in \Cont^\infty(M)$.

        \item It satisfies the following locality condition: for any $f,g \in \Cont^\infty(M)$,
        \begin{equation}
            \supp(\jacBr{f,g}) \subseteq \supp(f) \cap \supp(g),
        \end{equation}
        where $\supp(f)$ is the topological support of $f$, i.e., the closure of the set in which $f$ is non-zero.
    \end{enumerate}

    That is, $(\Cont^\infty(M), \jacBr{\cdot,\cdot})$ is a local Lie algebra in the sense of Kirillov~\cite{Kirillov1976}.
\end{definition}

Given a Jacobi manifold $(M,\Lambda, E)$ we can define a Jacobi bracket by setting
\begin{equation}\label{eq:jac_bracket_from_mfd}
    \jacBr{f,g} =\Lambda(\dd f, \dd g) + f E(g) - g E (f).
\end{equation}
In fact, every Jacobi bracket arises in this way.
\begin{theorem}\label{thm:jab_bra_characterization}
    Given a manifold $M$ and a $\RR$-bilinear map $\jacBr{\cdot,\cdot}: \Cont^\infty(M) \times \Cont^{\infty}(M) \to \Cont^\infty(M)$. The following are equivalent.
    \begin{enumerate}
        \item\label{item:jac_bra_local} The map $(\jacBr{\cdot,\cdot}$ is a Jacobi bracket.
        
        \item\label{item:jac_bra_leibniz} $(M,\jacBr{\cdot,\cdot})$ is a Lie algebra which satisfies the generalized Leibniz rule
        \begin{equation}\label{eq:mod_leibniz_rule}
            \jacBr{f,gh} = g\jacBr{f,h} + h\jacBr{f,g} +  g h E(h),
        \end{equation}
        where $E$ is a vector field on $M$.

        \item\label{item:jac_bra_mfd} There is a bivector field $\Lambda$ and a vector field $E$ such that $(M,\Lambda,E)$ is a Jacobi manifold and $\jacBr{\cdot,\cdot}$ is given as in \cref{eq:jac_bracket_from_mfd}.
    \end{enumerate}
\end{theorem}
\begin{proof}
    By a straightforward computation, (\ref{item:jac_bra_mfd}) implies (\ref{item:jac_bra_leibniz}).

    The statement (\ref{item:jac_bra_local}) follows from (\ref{item:jac_bra_leibniz}) by noticing that the generalized Leibniz rule implies that the map $X_f:\Cont^\infty(M) \to \Cont^\infty(M)$ such that $X_f (g) = \jacBr{f,g} +g E(f)$ is a $\RR$-linear derivation on $\Cont^\infty(M)$, hence it defines a smooth vector field. Therefore, if $g$ vanishes on a neighborhood of $p\in M$ then $X(g)$ and $g E(f)$ also vanish and, consequently, so does $\jacBr{f,g}$. Hence, $\supp(\jacBr{f,g}) \subseteq \supp(f) \cap \supp(g)$.

    In~\cite[Section~2]{Kirillov1976}, it was proven that every local Lie algebra on the space of functions is provided by a Jacobi structure, that is, (\ref{item:jac_bra_local}) implies (\ref{item:jac_bra_mfd}).
\end{proof}

We say that a function $f$ on a Jacobi manifold $M$ is \emph{a Casimir function} if $\jacBr{f,g} = 0$ for any other function $g$.

\begin{remark}
    For Poisson manifolds in which $E = 0$, the generalized Leibniz rule of the Jacobi brackets is the usual Leibniz rule and $\jacBr{\cdot,\cdot}$ are the so-called Poisson brackets.
    
    By noticing
    \begin{equation}
        \jacBr{1,f} = E(f),
    \end{equation}
    a Jacobi bracket is Poisson if and only if the constants are Casimir functions.
\end{remark}

We finish this chapter by noticing that the evolution of any observable $f$ under the flow of the Hamiltonian vector field of the function $H$ on $(M,\Lambda, E)$ can be written in terms of the brackets:
\begin{equation}\label{eq:Jacobi_evolution}
    \dot{f} = X_H(f) = \jacBr{H,f} + fE(H) =  \jacBr{H,f} + f \jacBr{1,H}.
\end{equation}

\section{The Lagrangian formalism}\label{sec:lagrangian}
In this section we will recall the geometric setting for time dependent Lagrangian systems, based on the cosymplectic geometry. We will emphasize the differences with the contact framework.

\subsection{Cosymplectic systems}
Let $L : TQ \times \mathbb R \to \mathbb R$ be a Lagrangian function, 
where $Q$ is the configuration $n$-dimensional manifold of a mechanical system. Then, $L = L(q^i, \dot{q}^i, z)$, where
$(q^i)$ are coordinates in $Q$, $(q^i, \dot{q}^i)$ are the induced bundle coordinates in $TQ$
and $z$ is a global coordinate in $\mathbb R$.

We will assume that $L$ is regular, that is, the Hessian matrix with respect to the velocities
\begin{equation}
	W_{ij}=\left( \frac{\partial^2 L}{\partial \dot{q}^i \partial \dot{q}^j} \right)
\end{equation}
is regular.

From $L$, and using the canonical endomorphism $S$ on $TQ$ locally defined by

\begin{equation}\label{eq:canonical_endomorphism}
    S = \dd q^i \otimes \frac{\partial}{\partial \dot{q}^i}    
\end{equation}
we can construct a 1-form $\lambda_L$ defined by
\begin{equation}
	\lambda_L = S^* (\dd L),
\end{equation}
where now $S$ and $S^*$ are the natural extensions of $S$ and its adjoint operator $S^*$ to $TQ \times \mathbb R$.
Therefore, we have
\begin{equation}
	\lambda_L = \frac{\partial L}{\partial \dot{q}^i}  \dd q^i.
\end{equation}

We will consider the cosymplectic structure $(\Omega_L, \dd z)$, where
\begin{equation}
	\Omega_L = - \dd \lambda_L.
\end{equation}

It is easy to check that, if $L$ is regular, then
\begin{equation}
	\dd z \wedge \Omega_L^n \not= 0,
\end{equation}
and conversely.

The Reeb vector field reads as follows: 
\begin{equation}
    \mathcal R = \frac{\partial}{\partial z} - 
    W^{ij} \frac{\partial^2 L}{\partial q^i \partial z} \pdv{}{\dot{q}^j},
\end{equation}
where $(W^{ij})$ is the inverse of the Hessian matrix of $L$ with respect to the velocities.

The energy of the system is defined by
\begin{equation}
    E_L = \Delta (L) - L,
\end{equation}
where $\Delta$ is the Liouville vector field.

\begin{equation}\label{eq:Liouvile_vf}
    \Delta = \dot{q}^i \, \frac{\partial}{\partial \dot{q}^i}
\end{equation}
such that
\begin{equation}\label{eq:energy_coords}
    E_L = \dot{q}^i \, \frac{\partial L}{\partial \dot{q}^i} - L.
\end{equation}

Consider now the following vector fields determined by means of the vector bundle isomorphism
\begin{align*}
\widetilde{\flat_L}  :  T(TQ \times \mathbb R) &\to T^*(TQ \times \mathbb R),\\
v &\mapsto \contr{v}  \Omega_L + \dd z (v)  \dd z.
\end{align*}
say,
\begin{enumerate}
\item the gradient vector field
\begin{equation}
	\grad (E_L) = \widetilde{\sharp_L} (dE_L),
\end{equation}
\item the Hamiltonian vector field
\begin{equation}
	X_{E_L} = {\mathcal E}_L - \mathcal R (E_L)  \mathcal R,
\end{equation}
\item and the evolution vector field
\begin{equation}
	{\mathcal E}_L = X_{E_L} + \mathcal R.
\end{equation}
\end{enumerate}
where $\widetilde{\sharp_L} = (\widetilde{\flat_L})^{-1}$ is the inverse of $\widetilde{\flat_L}$ (see \cite{Cantrijn1992}).

The evolution vector field ${\mathcal E}_L$ is locally given by
\begin{equation}\label{cosylagr1}
{\mathcal E}_ L = \dot{q}^i \, \frac{\partial}{\partial q^i} + B^i \, \frac{\partial}{\partial \dot{q}^i} + \frac{\partial}{\partial z},
\end{equation}
where
\begin{equation}\label{cosylagr2}
B^i \, \frac{\partial}{\partial \dot{q}^i}
  \parens*{\frac{\partial L}{\partial \dot{q}^j}} + \dot{q}^i  
\frac{\partial}{\partial q^i} \parens*{\frac{\partial L}{\partial \dot{q}^j}} - \frac{\partial L}{\partial q^j} = 0.
\end{equation}
Now, if $(q^i(t), \dot{q}^i(t), z(t))$ is an integral curve of ${\mathcal E}_L$, then it satisfies the 
usual Euler-Lagrange equations
\begin{equation}\label{cosylagr3}
\frac{\dd}{\dd t} \parens*{\frac{\partial L}{\partial \dot{q}^i}} - \frac{\partial L}{\partial q^i} = 0.
\end{equation}
since $z = t + \text{constant}$.

Notice that ${\mathcal E}_L$ is a second order differential equation (SODE, for short) because $S({\mathcal E}_L) = \Delta$, reflecting the fact that the Euler-Lagrange equations are second order.

\subsection{Contact systems}

In this case, we also have a regular Lagrangian $L : TQ \times \mathbb R \to \mathbb R$,
but instead to consider the cosymplectic structure $(\Omega_L, \dd z)$, we will consider the contact structure given by the $1$-form

\begin{equation}\label{eq:Lagrangian_form}
    \eta_L = \dd z - \lambda_L,    
\end{equation}
which is a contact form on $TQ \times \mathbb R$ if and only if $L$ is regular; indeed,
if $L$ is regular, then
\begin{equation}
	\eta_L \wedge (\dd \eta_L)^n \not= 0,
\end{equation} 
and conversely. The corresponding Reeb vector field is, again,
\begin{equation}
    \Reeb = \pdv{}{z} - W^{ij} \frac{\partial^2 L}{\partial q^i \partial z} \pdv{}{\dot{q}^j},
\end{equation}
where $(W^{ij})$ is the inverse of the Hessian matrix of $L$ with respect to the velocities.

The energy of the system is defined just like in the cosymplectic case:
\begin{equation}
    E_L = \Delta (L) - L,
\end{equation}
where $\Delta$ is the Liouville vector field on $TQ$ extended in the usual way to $TQ \times \RR$.

Denote by
\begin{equation}
	\bar{\flat}_L : T(TQ \times \mathbb R) \to T^* (TQ \times \mathbb R)
\end{equation}
the vector bundle isomorphism
\begin{equation}
	\bar{\flat}_L (v) = \contr{v} (\dd\eta_L) + (\contr{v} \eta_L) \, \eta_L
\end{equation}
given by the contact form $\eta_L$ on $TQ \times \mathbb R$.
We shall denote its inverse by $\bar{\sharp}_L = (\bar{\flat}_L)^{-1}$.

Denote by $\bar{\xi}_L$ the unique vector field defined by the equation
\begin{equation}\label{clagrangian1}
\bar{\flat}_L (\bar{\xi}_L) = dE_L - (\mathcal R(E_L) + E_L) \, \eta_L
\end{equation}

A direct computation from eq. (\ref{clagrangian1}) shows that $\bar{\xi}_L$ is locally given by
\begin{equation}\label{eq:lagrangian_solution}
\bar{\xi}_L = \dot{q}^i  \frac{\partial}{\partial q^i} + {b}^i \,\frac{\partial}{\partial \dot{q}^i} 
+ L \frac{\partial}{\partial z},
\end{equation}
where the components ${b^i}$ satisfy the equation
\begin{equation}\label{clagrangian3}
{b}^i  \frac{\partial}{\partial \dot{q}^i}\parens*{\frac{\partial L}{\partial \dot{q}^j}} 
+ \dot{q}^i  \frac{\partial}{\partial q^i}\parens*{\frac{\partial L}{\partial \dot{q}^j}} 
+ {L  \frac{\partial}{\partial z}\parens*{\frac{\partial L}{\partial \dot{q}^j}}} - \frac{\partial L}{\partial q^j} =
\frac{\partial L}{\partial \dot{q}^j} \frac{\partial L}{\partial z}
\end{equation}

Then, if $(q^i(t), \dot{q}^i(t), z(t))$ is an integral curve of $\bar{\xi}_L$ and substituting its values in \cref{clagrangian3}, we obtain
\begin{equation}
    {\ddot{q}}^i  \frac{\partial}{\partial \dot{q}^i}\parens*{\frac{\partial L}{\partial \dot{q}^j}} 
    + \dot{q}^i  \frac{\partial}{\partial q^i}\parens*{\frac{\partial L}{\partial \dot{q}^j}} 
    +  \dot{z}  \frac{\partial}{\partial z}\parens*{\frac{\partial L}{\partial \dot{q}^i}} - \frac{\partial L}{\partial q^i} =
    \frac{\partial L}{\partial \dot{q}^i} \frac{\partial L}{\partial z}
\end{equation}
which corresponds to the generalized Euler-Lagrange equations considered by G. Herglotz in 1930~\cite{Herglotz1930} (see also  \cite{Georgieva2003,Georgieva2011})
\begin{equation}\label{clagrangian4}
\frac{\dd}{\dd t} \parens*{\frac{\partial L}{\partial \dot{q}^i}} - \frac{\partial L}{\partial q^i} =
\frac{\partial L}{\partial \dot{q}^i} \frac{\partial L}{\partial z}.
\end{equation}

Notice that $\bar{\xi}_L$ is a SODE, that is, $S(\bar{\xi}_L) = \Delta$.

\section{The Hamiltonian formalism and the Legendre transformation}\label{sec:Hamiltonian_Legendre_transform}
In this section we will discuss the Hamiltonian description an how it is related to the Lagrangian description via the Legendre transformation. We will use the results an notations from \cref{sec:geometry_and_dynamics}.

\subsection{The Hamiltonian formalism}

Let $H : T^*Q \times \RR \to \RR$ be a Hamiltonian function, say
$H = H(q^i, p_i, z)$ where $(q^i, p_i, z)$  are bundle coordinates on $T^*Q \times \mathbb R$.
Consider the 1-form
\begin{equation}
	\eta = dz - \theta_Q 
\end{equation}
where $\theta_Q$ is the canonical Liouville form on $T^*Q$. In what follows we will consider
the usual identification for a form on $T^*Q$ or $\mathbb R$ and its pull-back to $T^*Q \times \mathbb R$.
In local coordinates, we have
\begin{equation}
	\eta = dz - p_i \, \dd q^i
\end{equation}
So, $\eta$ is a contact form on $T^*Q \times \mathbb R$ and $(q^i, p_i, z)$ are Darboux coordinates.
Therefore, we can obtain a Hamiltonian vector field $X_H$ which locally takes the same form that in \cref{eq:contact_hamiltonian_vf_darboux}.

\subsection{Legendre transformation}

Given a Lagrangian function $L : TQ \times \mathbb R \to \mathbb R$ we can define the Legendre transformation 
\begin{equation}
	\Fib L : TQ \times \mathbb R \to T^*Q \times \mathbb R
\end{equation}
given by
\begin{equation}\label{eq:legendre_transform_contact}
    \begin{aligned}
        \Fib L (q^i, \dot{q}^i, z) &= (q^i, \hat{p}_i, z) \\
        \hat{p}_ i &= \frac{\partial L}{\partial \dot{q}^i}
    \end{aligned}
\end{equation}

We will assume that $L$ is \emph{hyperregular}, that is, the Legendre transformation is a diffeomorphism. Consequently, the generalized Euler-Lagrange equations are transformed into the contact Hamilton equations.

Indeed, a direct computation shows that
\begin{equation}
	{(\Fib L)}^* \eta = \eta_L,
\end{equation}
and then we have
\begin{equation}
	(\Fib L)_{*}(\tilde{\xi}_L) = X_H,
\end{equation}
where $H =E_L \comp (\Fib L)^{-1}$.

The Legendre transformation is the same for the cosymplectic and contact settings. Therefore, $\Fib L$ also connects the corresponding Lagrangian and Hamiltonian formalisms in the obvious manner~\cite{Chinea1994}.

\section{Variational formulation of contact Lagrangian mechanics}\label{sec:variational}

Let $L:TQ \times \RR \to \RR$ be a Lagrangian function. In this section we will recall the so-called Herglotz's principle, a modification of Hamilton's principle that allows us to obtain Herglotz's equations (\cref{clagrangian4}), sometimes called generalized Euler-Lagrange equations. See~\cite{Herglotz1930}, or \cite{Liu2018} for a recent discussion.

Fix $q_1,q_2 \in Q$ and an interval $[a,b] \subset \RR$. We denote by $\Omega(q_1,q_1, [a,b]) \subseteq(\Cont^\infty([a,b]\to Q))$ the space of smooth curves $\xi$ such that $\xi(a)=q_1$ and $\xi(b)=q_2$. This space has the structure of an infinite dimensional smooth manifold whose tangent space at $\xi$ is given by the set of vector fields over $\xi$ that vanish at the endpoints~\cite[Proposition~3.8.2]{Abraham1978}, that is,
\begin{equation}
\begin{aligned}
        T_\xi \Omega(q_1,q_2, [a,b]) =  \set{&
            v_\xi \in \Cont^\infty([a,b] \to TQ) \mid \\& 
            \tau_Q \comp v_\xi = \xi, \,v_\xi(a)=0, \, v_\xi(b)=0 
            }.
\end{aligned}
\end{equation}

We will consider the following maps. Fix $c \in \RR$. Let 
\begin{equation}
    \mathcal{Z}:[a,b] \to \Cont^\infty ([a,b] \to Q)
\end{equation}
 be the operator that assigns to each curve $\xi$ the curve $\mathcal{Z}(\xi)$ that solves the following implicit ODE:
\begin{equation}\label{contact_var_ode}
    \dv{\mathcal{Z}(\xi)(t)}{t} = L(\xi(t), \dot \xi(t), \mathcal{Z}(\xi)(t)), \quad \mathcal{Z}(\xi)(a)= c.
\end{equation}

Now we define the \emph{action functional} as the map which assigns to each curve the solution to the previous ODE evaluated at the endpoint:
\begin{equation}\label{contact_action}
    \begin{aligned}
        \mathcal{A}: \Omega(q_1,q_2, [a,b]) &\to \RR,\\
        \xi &\mapsto \mathcal{Z}(\xi)(b),
    \end{aligned}
\end{equation}
that is, $\mathcal{A} = \ev_b \comp \mathcal{Z}$, where $\ev_b: \zeta \mapsto \zeta(b)$ is the evaluation map at $b$.

\begin{theorem}[Contact variational principle]
    Let $L: TQ \times \RR \to \RR$ be a Lagrangian function and let $\xi\in  \Omega(q_1,q_2, [a,b])$ be a curve in $Q$. Then, $(\xi,\dot\xi, \mathcal{Z}(\xi))$ satisfies the Herglotz's equations (\cref{clagrangian4}) if and only if $\xi$ is a critical point of $\mathcal{A}$.
\end{theorem}
\begin{remark}
    This theorem generalizes Hamilton's Variational Principle~\cite[Theorem~3.8.3]{Abraham1978}. In the case that the Lagrangian is independent of the $\RR$ coordinate (i.e., $L(x,y,z)=\hat{L}(x,y)$)  the contact Lagrange equations reduce to the usual Euler-Lagrange equations. In this situation, we can integrate the ODE of \cref{contact_action} and we get
    \begin{equation}
        \mathcal{A}(\xi) = \int_a^b \hat L(\xi(t),\dot\xi(t))\dd t + \frac{c}{b-a},
    \end{equation}
    that is, the usual Euler-Lagrange action up to a constant.
\end{remark}

\begin{proof}
    Let $\xi \in \Omega(q_1,q_2, [a,b])$ be a curve and consider some tangent vector $v \in T_\xi \Omega(q_1,q_2, [a,b])$. In order to simplify the notation, let $\chi=(\xi,\dot\xi, \mathcal{Z}(\xi))$, which is a curve in $TQ \times \RR$. We will compute $T_{\xi} \mathcal{A}(v)$. Consider a smoothly parametrized family of curves $\xi_\lambda$ in $\Omega(q_1,q_2, [a,b])$, such that
    \begin{equation*}
        v=\restr{\dv{\xi_\lambda}{\lambda}}{\lambda=0}
    \end{equation*} 
    and let $\psi = T_\xi \mathcal{Z}(v)$, so that $T_{\xi} \mathcal{A}(v)= \psi(b)$. We first notice that, since $\mathcal{Z}(\xi_\lambda)(0)=c$ for all $\lambda$, we have that $\psi(0)=0$. Now, we will compute the time derivative of $\psi$. By interchanging the order of the derivatives and using Darboux coordinates, we find out that
    \begin{align*}
        \dot{\psi}(t) &= 
        \restr{\dv{}{\lambda}\dv{}{t} 
        \mathcal{Z}(\xi_\lambda(t))}{\lambda=0} \\&=
        \dv{}{\lambda}\restr{L(\xi_\lambda(t), \dot{\xi}_\lambda(t),\mathcal{Z}(\xi_\lambda)(t))}{\lambda=0} \\&=
        \pdv{L}{x}(\chi(t)) v(t) +
        \pdv{L}{\dot{x}}(\chi(t)) \dot{v}(t) + 
        \pdv{L}{z}(\chi(t)) \psi(t).
    \end{align*}

Hence, the function $\psi$ is the solution to the ODE above. Explicitly
\begin{equation}
    \sigma(t) \psi(t) = \int_a^t \sigma(\tau) \parens*{
        \pdv{L}{x}(\chi(\tau)) v(\tau) + \pdv{L}{\dot{x}}(\chi(\tau)) \dot{v}(\tau)
        } \dd \tau,
\end{equation}
where,
\begin{equation}
    \sigma(t) = \exp \parens*{-\int_a^t \pdv{L}{z}(\xi(\tau)) \dd \tau} > 0.
\end{equation}

Since $v(a)=v(b)=0$, we can integrate by parts and find out that
\begin{gather*}
    \sigma(b) \psi(b) = 
    \int_a^b  v(t) \parens*{
        \sigma(t) \pdv{L}{x}(\chi(t)) - 
        \dv{}{t} \parens*{\sigma(t)\pdv{L}{\dot{x}}(\chi(t))} 
        } \dd t \\ =
        \int_a^b  v(t) \sigma(t) \parens*{
         \pdv{L}{x}(\chi(t)) - 
        \dv{}{t} \pdv{L}{\dot{x}}(\chi(t))
        +  \pdv{L}{\dot{x}}(\chi(t)) \pdv{L}{z}(\chi(t)) 
        } \dd t.
\end{gather*}
Since  $\sigma(t)$ is nonzero, by the fundamental lemma of the calculus of variations, $\psi(b) = 0$ for every $v$ (i.e., $\chi$ is a critical point of $\mathcal{A}$) if and only if
\begin{equation}
    \pdv{L}{x}(\chi(t)) - 
    \dv{}{t} \pdv{L}{\dot{x}}(\chi(t))
    +  \pdv{L}{\dot{x}}(\chi(t)) \pdv{L}{z}(\chi(t)) = 0.
\end{equation}
\end{proof}

\begin{remark}
    By the results of \cref{sec:lagrangian}, if the Lagrangian is regular, then Herglotz's equations, (and, therefore, the variational problem) is equivalent to a contact Hamiltonian system. However this is not true for general Lagrangians. In the following chapters we will provide some tools to deal with singular Lagrangians.
\end{remark}

\section{Precontact manifolds}\label{sec:precontact_maifolds}
The theory presented on the previous section provides well defined dynamics for regular Lagrangian systems and there is a satisfactory correspondence between the Lagrangian and Hamiltonian formalisms on the hyperregular case. However, we would like to treat more general kinds of systems in which Lagrangians are allowed to be singular. For that, we will need to introduce a geometric model that generalizes contact geometry: precontact geometry. This geometry plays a similar role than presymplectic geometry for singular symplectic Lagrangian systems.

Let $\eta$ be a $1$-form in an $m$-dimensional manifold $M$. We define the \emph{characteristic distribution} of $\eta$ as
\begin{equation}
    \CharD= \ker \eta \cap \ker \dd \eta \subseteq TM.
\end{equation}
We say that $\eta$ is \emph{of class} $c$ if $\CharD$ is a distribution of rank $m-c$.

\begin{proposition}\label{thm:class_of_form}
    Let $\eta$ be a one-form on an $m$-dimensional manifold $M$. Then it is equivalent:
    \begin{enumerate}
        \item The form $\eta$ is of class $2r+1$.
        \item At every point of $M$,
        \begin{equation}
            \eta \wedge {(\dd \eta)}^r \neq 0, \quad 
            \eta \wedge {(\dd \eta)}^{r+1} = 0.
        \end{equation}
        \item Around any point of $M$, there exist local \emph{Darboux} coordinates $x^1,\ldots x^r$, $y_1, \ldots y_r$, $z$, $u_1, \ldots u_s$, where $2r+s+1 = m$, such that
        \begin{equation}
            \eta = \dd z - \sum_{i=1}^r y_i \dd x^i.
        \end{equation}
    \end{enumerate}
\end{proposition}
For a proof, see \cite[Theorems~VI.1.6, VI.4.1]{Godbillon1969}. In that situation we say that $\eta$ is \emph{a precontact form of class $2r+1$}. In coordinates, the characteristic distribution is given by
\begin{equation}
    \CharD = \gen*{\set*{\pdv{}{u_a}}_{a=1,\ldots,s}}.
\end{equation}

A pair $(M,\eta)$ of a manifold equipped with a precontact form will be called a \emph{precontact manifold}. A triple $(M,\eta,H)$, where $(M,\eta)$ is a precontact manifold and $H\in \Cont^\infty (M)$ is the \emph{Hamiltonian function} will be called a \emph{precontact Hamiltonian system}, which is the main object of study in this article.

\begin{remark}
    The distribution $\CharD$ is involutive and it gives rise to a foliation of $M$. If the quotient $\pi:M \to M/\CharD$ has a manifold structure, then there is a unique $1$-form $\tilde\eta$ such that 
    $\pi^* \tilde{\eta}=\eta$. From a direct computation, $\tilde{\eta}$ is a contact form on $M/\CharD$. This justifies the name of \emph{precontact form}.
\end{remark}

    We define the following morphism of vector bundles over $M$:
    \begin{equation}
        \begin{aligned}
            \bar{\flat}: TM &\to TM^*\\
            v &\mapsto \contr{v} \dd \eta + \eta(v) \eta.
        \end{aligned}
    \end{equation}
    
    The following $2$-tensors are associated to $\bar{\flat}$ and its transpose
    \begin{equation}
        \omega = \dd \eta + \eta \otimes \eta, \quad \bar\omega = - \dd \eta + \eta \otimes \eta.
    \end{equation}
    In other words, $\bar{\flat}(X) = \omega(X,\cdot) = \bar\omega(\cdot,X)$. Therefore $\omega(X,Y) = \bar\omega(Y,X)$.
    
    A \emph{Reeb vector field} for $(M,\eta)$ is a vector field $\Reeb\in \VecFields (M)$ such that
    \begin{equation}
        \contr{\Reeb} \dd \eta = 0, \quad \eta(\Reeb) = 1. 
    \end{equation}
    
    We note that there exists Reeb vector fields in every precontact manifold. Indeed we can define local vector fields $\Reeb = \pdv{}{z}$ in Darboux coordinates and can extend it using partitions of unity.

    \begin{proposition}
        Let $(M,\eta)$ be a precontact manifold. We have  
        \begin{equation}\label{eq:CharD}
            \CharD =  \ker \eta \cap \ker \dd \eta = \ker \bar{\flat} = \ann{(\im \bar{\flat})}.
        \end{equation}
    \end{proposition}
    \begin{proof}
        We will prove the previous equalities. In order to see that $\ker \eta \cap \ker \dd \eta = \ker \bar{\flat}$, let $\bar{\flat}(X)=0$, then $\contr{X} \dd \eta + \eta(X) \eta = 0$. If we contract the previous expression with a Reeb vector field $\Reeb$ we obtain that $\eta(X) = 0$. Thus, $\contr{X} \dd \eta$ also vanishes. The other inclusion is trivial.
    
        Now we will see that $\ann{(\im \bar{\flat})} =  \ker \bar{\flat}$. Let $X\in \ker \bar{\flat}$. By the first equality, $\contr{X}\dd\eta=0$ and $\contr{X}\eta=0$. Then, for any vector field $Y$
        \begin{equation}
            \contr{X}\bar{\flat}(Y) = \contr{X} \contr{Y} \dd\eta + \eta(Y)\eta(X) = -\contr{Y}\contr{X}\dd\eta = 0,
        \end{equation}
        hence $\ann{(\im \bar{\flat})} \supseteq  \ker \bar{\flat}$. By noticing that at each point $p\in M$ both subspaces of $T_p M$ have the same dimensions, we conclude that both distributions are equal.
    \end{proof}

    \begin{proposition}\label{Reeb_im_eta}
        A vector field $X$ is a Reeb vector field for $(M,\eta)$ if and only if $\bar{\flat}(X) = \eta$. That is, the set of Reeb vector fields is $\Reeb + \sects(\CharD)$, where $\Reeb$ is an arbitrary Reeb vector field and $\sects(\CharD)$ is the set of vector fields tangent to $\CharD$.
    \end{proposition}
    \begin{proof}
        Let $\Reeb$ be a Reeb vector field. Then, $X$ is also a Reeb vector field if and only if $\eta(X)=\eta(\Reeb) = 1$ and 
        $\contr{X} \eta = \contr{{\Reeb}} \eta=0$. That is, if and only if $\Reeb-X$ is tangent to $\CharD$. Equivalently $\bar{\flat}(\Reeb-X)=0$ or $\bar{\flat}(X)=\eta$.
    \end{proof}
    
    For a distribution $\Delta \subseteq TM$, we define the following notion of complement with respect to $\omega$. Since $\omega$ is neither symmetric nor antisymmetric, we need to distinguish between right and left complements:
    \begin{equation}
        \begin{aligned}        
            \orth{\Delta}  &=  
              \set{X \in TM \mid \omega(Z,X) = \bar{\flat}(Z)(X) = 0,\, \forall Z\in \Delta}
              = \ann{(\bar{\flat}(\Delta))},\\
            \orthL{\Delta} &= 
              \set{X \in TM\mid \omega(X,Z) = 0, \forall Z \in \Delta}.
        \end{aligned}
    \end{equation}
    These complements have the following relationship
    \begin{equation}\label{complements_cancel}
        \orthL{(\orth{\Delta})} = 
        \orth{(\orthL{\Delta})} = \Delta + \CharD.
    \end{equation}
    We remark that these complements interchange sums and intersections, since the annihilator interchanges them and the linear map $\bar{\flat}$ preserves them. Consequently, if $\Delta,\Gamma$ are distributions, we have
    \begin{equation}
        \begin{aligned}
            \orth{(\Delta\cap \Gamma)} &= \orth{\Delta} + \orth{\Gamma}\\
            \orth{(\Delta +  \Gamma)} &= \orth{\Delta}  \cap \orth{\Gamma}
        \end{aligned}
    \end{equation}

    \section{The constraint algorithm}\label{sec:constr_alg}
    We aim to solve Hamilton equations on a precontact Hamiltonian system $(M, \eta, H)$. In order to do that, we will introduce an algorithm similar to the one introduced on~\cite{Gotay1978} for presymplectic systems and that was extended in~\cite{Chinea1994,Leon2002} to the cosymplectic case.
    
    Let $\gamma_H = \dd H - (H + \Reeb(H))\eta$ where $\Reeb$ is a Reeb vector field (we will later see that the algorithm is independent on the choice of the Reeb vector field) and consider the equation
    \begin{equation}
        \bar{\flat}(X) = \gamma_H.
    \end{equation}
    This equation might not have solution, so we will consider the subset $M_1 \subseteq M_0 = M$ of the points in which a solution exists. That is,
    \begin{equation}
        M_1 = \set{p \in M_0 \mid (\gamma_H)_p \in \bar{\flat}(T_p M_0) }.
    \end{equation}
    We note that this condition is equivalent to the following
    \begin{equation}
        M_1 = \set{p \in M_0 \mid \bracks{(\gamma_H)_p ,\orth{TM_0}} = 0},
    \end{equation}
    since $\bar{\flat}(TM_0) = \ann{\parens{\ann{\bar{\flat}(TM_0)}}}=\ann{\parens{\orth{TM_0}}}$.
    
    If we choose a local basis $\set{X_a}_{j=a}^{k_1}$ of $\orth{TM_0}$,  we can easily compute the so-called \emph{primary constraint} functions $\phi^a(p) = \bracks{\dd H_p - (\Reeb(H) + H) \eta_p ,X_a}$, whose zero set is the manifold $M_1$. We note that $\orth{TM_0} =\ann{(\im \bar{\flat})} = \ker \bar{\flat} = \CharD$ by \cref{eq:CharD}. Hence,
    \begin{equation}
        \bracks{\dd H_p - (\Reeb(H) + H) \eta_p ,\orth{TM_0}}= 
        \set{Z_p(H) = 0 \mid Z_p \in \CharD_p}.
    \end{equation}
    Therefore, in Darboux coordinates,
    \begin{equation}
        \phi^a = \pdv{H}{s^a}.
    \end{equation}
    We note that this implies that $\Reeb = \tilde\Reeb(H)$ along $M_1$ for every Reeb vector field $\tilde\Reeb$, since $\Reeb_p-\tilde{\Reeb}_p \in \CharD_p$. Consequently, $\restr{\gamma_H}{M_1}$ is independent on the choice of the Reeb vector field. Therefore, the election of $\Reeb$ doesn't affect the constraints produced by the algorithm.

    Now we can solve Hamilton equations, but, in order to have meaningful dynamics, the solution $X$ should be tangent to the constraint submanifold. Otherwise, a solution of the equations of motion might escape from $M_1$. This tangency condition is equivalent to demand that $\bar{\flat}(X_p) \in \bar{\flat}(TM_p)$ since $\bar{\flat}$ is an isomorphism modulo $\CharD_p$: 
    \begin{equation}
        M_2 = \set{p \in M_1 \mid \bracks{(\gamma_H)_p ,\orth{TM_1}} = 0},
    \end{equation}
    providing a second constraint submanifold, with its corresponding constraint functions. However, it is not enough. We must again require that the vector field is tangent to the new submanifold. We then get a sequence of submanifolds
    \begin{equation}
    \begin{aligned}
            M_{i+1} &= 
            \set{p \in M_i \mid (\gamma_H)_p \in\bar{\flat}(T_p M_i)} \\ &=
            \set{p \in M_i \mid \bracks{(\gamma_H)_p ,\orth{T_p M_i}} = 0}    
        \end{aligned}
    \end{equation}
    which eventually stabilizes, that is, there exist some $i_f$ such that $M_{i_f} = M_{i_f+1}$. We call this manifold the \emph{final constraint submanifold} and denote it by $M_f$. This submanifold is locally described by the zero set of some constraint functions $\set{\phi^j}_{j=1}^{k_f}$.

    \subsection{Tangency of the Reeb vector field}
    Next, we will discuss when there is a Reeb vector field tangent to the final constraint submanifold. We can guarantee it in some situations, like in the case of Rayleigh dissipation (as in the example of \cref{ssec:ex1}) in which $\Reeb(H)$ is constant. However, this is not true in general, as can be seen in the example of \cref{ssec:ex_no_reeb}.

        \begin{lemma}\label{thm:reeb_horizontal_complement}
            Let $N$ be a submanifold of a precontact manifold $(M, \eta)$. Then, there exists a Reeb vector field on $M$ tangent to $N$ if and only if $\orth{TN} \subseteq \ker \eta$
        \end{lemma}
        \begin{proof}
           Let $\Reeb$ be a Reeb vector field on $M$ tangent to $N$. Let $X$ be tangent to $\orth{TN}$. That is, for all $Y_p \in \orth {TN}$ and $p \in N$,
            \begin{equation}
                \bar{\flat}{(Y_p)}(X_p) = 0.
            \end{equation}
            In particular, if we let $Y=\Reeb$,
            \begin{equation}
                \bar{\flat}{(\Reeb_p)}(X_p) = 
                \dd \eta(\Reeb_p,X_p) + \eta(\Reeb_p) \eta(X_p) 
                 = \eta(X_p) = 0,
            \end{equation}
            hence $\orth{TN} \subseteq \ker \eta$.

            For the converse, $\orth{TN} \subseteq \ker \eta$ implies $\eta \in \ann{(\orth{TN})} = \bar{\flat}(TN)$. So there $\eta = \bar{\flat}(Y)$ with $Y$ tangent to $TN$. $Y$ is a Reeb vector field by \cref{Reeb_im_eta}.
        \end{proof}
        
        \begin{proposition}\label{thm:reeb_tangent}
            Let $(M,\eta, H)$ be a precontact Hamiltonian system. Then, there is a Reeb vector field $\Reeb$ tangent to the final constraint submanifold if and only if $Z(\Reeb (H)) = 0$ for all $Z \in \orth{TM_f}$. In particular, if $\Reeb(H)$ is constant, then $\Reeb$ is tangent to $M_f$.
        \end{proposition}
        \begin{proof}
            We will prove the result by induction. Let $p\in TM_f$ and let $\Reeb$ be a Reeb vector field tangent to $M_i$. Notice that $\orth{TM_i}\subseteq (\ker \eta)$ by \cref{thm:reeb_horizontal_complement}. $\Reeb$ will be tangent to $M_{i+1}$ at $p$ if the Lie derivative of the $(i+1)$-th constraint functions vanish.
            That is, for every $Z$ tangent to $\orth{T M_i}$ in a neighborhood of $p$,
            \begin{equation}\label{eq:Reeb_restriction}
                (\lieD{\Reeb} \bracks{\gamma_H,Z})_p= \bracks{\lieD{\Reeb} \gamma_H , Z}_p + \bracks{\gamma_E, \lieD{X} Z}_p = 0.
            \end{equation}

            We compute the first term. Since $\lieD{\Reeb} \eta = 0$, we have that
            \begin{align*}
                \lieD{\Reeb} \gamma_H &= 
                \lieD{\Reeb} \dd H - \lieD{X} (E + \Reeb(E)) \eta .
            \end{align*}
            Therefore, because $\eta(Z)=0$, we deduce
            \begin{equation}
                \bracks{\lieD{\Reeb} \gamma_H , Z} = Z(\Reeb(H))
            \end{equation}
            
            We will now see that $\lieBr{R,Z}_p \in T_p M_i$. Let $W$ be any vector field on $M_i$. Then, along $M_i$,
            \begin{align*}
                \omega(W,\lieBr{\Reeb,Z}) =  
                - \lieD{\Reeb} \omega(W,Z) +
                \lieD{\Reeb}(\omega(W,Z)) + \omega(\lieBr{\Reeb,W},Z) = 0.
            \end{align*}                    
            The first term vanishes since $\lieD{\Reeb} \omega = 0$ because $\lieD{\Reeb} \eta = 0$. The second and third terms are also zero because $Z \in \orth{TM_i}$. Hence, the last term of \cref{eq:Reeb_restriction} vanishes along $M_{i+1}$. Therefore, $\Reeb$ is tangent to $M_{i+1}$ if and only if
            \begin{equation}
                (\lieD{\Reeb}\bracks{\gamma_H , Z})_p = Z(\Reeb(H))_p =0,
            \end{equation}
            for all $Z_p \in \orth{(TM_i)_p} \subseteq \orth{(TM_f)_p}$. 

            Notice that if $\Reeb$ is not tangent to $T_p M_{i}$ it will not be tangent to $TM_f \subseteq T_p M_{i}$, so the converse follows.
        \end{proof}

    As we have proven in \cref{thm:reeb_tangent}, the Reeb vector field is not necessarily tangent to the constraint submanifold $M_f$. A modification of the previous algorithm guarantees this fact, just by requiring that a chosen Reeb vector field $\Reeb$ is tangent to the constraint submanifold after each step. This will produce a new sequence of submanifolds. Explicitly, $\bar{M}_0 = \hat{M}_0 = M$, and for $i\geq 1$ we define recursively:
    \begin{equation}
        \begin{aligned}
            \bar{M}_{i} &= 
            \set{p \in \hat{M}_{i-1} \mid (\gamma_H)_p \in\bar{\flat}(T_p \hat{M}_{i-1})} \\ &=
            \set{p \in \hat{M}_{i-1}  \mid \bracks{(\gamma_H)_p ,\orth{T_p \hat{M}_{i-1}}} = 0}  \\
            \hat{M}_{i} &= 
            \set{p \in \bar{M}_i \mid \Reeb_p \in T_p \bar{M}_i} \\ &=
            \set{p \in \bar{M}_i \mid \lieD{\Reeb}\bracks{(\gamma_H)_p ,\orth{T_p \hat{M}_{i-1}}} = 0}.  
        \end{aligned}
    \end{equation}
    Locally, in terms of constraint functions, if $\bar{M}_i$ is described as the zero set of functions $(\phi_k)_k$, then $\hat{M}_i$ would be the zero set of $(\phi_k, \Reeb(\phi_k))_k$. We get a sequence of constraint submanifolds as follows:
    \begin{equation}
        \cdots \toinj \hat{M}_{i+1} \toinj \bar{M}_i \toinj \hat{M}_i \toinj \cdots \toinj \hat{M}_2 \toinj \bar{M}_1 \toinj \hat{M}_1 \toinj M,
    \end{equation}

    The algorithm stops when we reach submanifold such that none of the two steps produces new constraints. That is: $\bar{M}_{j_f} = \hat{M}_{j_f} = \bar{M}_{j_f+1}$.

    \begin{remark}
        By construction, the first algorithm will produce the largest submanifold $M_f$ in which there is a solution to the equations of motion. The second algorithm produces a final constraint submanifold $\bar{M}_{f}$ in which there is a solution to the equations of motion and a Reeb vector field is tangent, hence $\bar{M}_{f} \subseteq M_f$. Apart from this, no much more about the relationship between $M_f$ and $\bar{M}_{f}$ seems possible to state. It can be the case that $\bar{M}_{f} = M_f$, such as in the example of \cref{ssec:ex1}, or that $\bar{M}_{f}= \emptyset$ and $M_f$ is nonempty for any choice of Reeb vector field, as in the example of \cref{ssec:ex_no_reeb}.
    \end{remark}

    \section{The constraint algorithm and the Legendre transformation}\label{sec:equiv_problem}
    In this section, we will apply the previous constraint algorithm to singular Lagrangian systems.

    Let that $L: TQ \times \RR \to \RR$ is a singular Lagrangian function. We will use the results and notation of \cref{sec:lagrangian,sec:Hamiltonian_Legendre_transform}. As in the previous sections, we will denote by $\eta_L = \dd z - \lambda_L$ the $1$-form defined in \cref{eq:Lagrangian_form}; by $\Delta$ is the extended Liouville vector field (\cref{eq:Liouvile_vf}); by $S$ the canonical endomorphism (\cref{eq:canonical_endomorphism}, and by $\Fib L$ the Legendre transformation (\cref{eq:legendre_transform_contact}) with respect to $L$.

    The objective is twofold: to develop a constraint algorithm in the Lagrangian side, but also the corresponding Hamiltonian counterpart. 
    
    We make the following observation, which is useful for working with precontact systems that come from a Lagrangian. The proof is trivial from the coordinate expression of $\dd \eta_L$.
    \begin{proposition}
        Let $L:TQ\times\RR \to \RR$ be a Lagrangian function. Then, the form $\eta_L$ is precontact of class $2r+1$ if and only if the rank of the Hessian matrix of $L$ with respect to the velocities is $r$ at every point.
    \end{proposition}
    
    Let $E_L = \Delta(L)-L$ be the energy and $\gamma_{E_L}= \dd E_L - (\Reeb(E_L) + E_L)\eta_L$, where $\eta_L$ is a precontact form of class $2r+1$. We remark that $(TQ \times \RR,\eta_L,E_L)$ is a precontact Hamiltonian system. Hence, we can apply the constraint algorithm developed in \cref{sec:constr_alg} to the equation $\bar{\flat}_L (X) = \gamma_{E_L}$. 
    
    If we denote $P_1 = TQ \times \RR$, we will obtain a sequence of constraint submanifolds
    \begin{equation}
        \cdots \toinj P_{i}  \toinj \cdots \toinj P_2 \toinj P_1,
    \end{equation}
    where
    \begin{equation}
        P_{i+1} =
        \set{p \in P_i \mid \bracks{(\gamma_H)_p ,\orth{T_p P_i}} = 0}, 
    \end{equation}
    and $P_f$ is the final constraint submanifold. If it has positive dimension, then there would exist a vector field $X$ tangent to $P_f$ that solves the equations of motion along $P_f$.

    Of course, this solution will not be unique in general. We would get a new solution by adding a section of $\CharD \cap TP_f$, where $\CharD = \ker \bar{\flat}_L$ is the characteristic distribution.

\subsection{The Hamiltonian side and the equivalence problem}
Now we will develop a Hamiltonian counterpart of this theory. This problem was addressed in~\cite{Gotay1979} for singular Lagrangians in the presymplectic case and by~\cite{Chinea1994} for the time dependent case.     
We will require the following additional regularity conditions on $L$ to make sure we get a precontact Hamiltonian system which is amenable to the constraint algorithm:
    \begin{definition}
        We say that a contact Lagrangian
        $L\in \Cont^{\infty}(TQ \times \RR)$ is \emph{almost regular} if
        \begin{itemize}
            \item $\eta_L$ is precontact.
            \item $\Fib L$ is a submersion onto its image.
            \item For every $p\in T^*Q \times \RR$, the fibers ${(\Fib L)}^{-1}(p)$ are connected submanifolds. 
        \end{itemize}
    \end{definition}

    We denote by $M_1$ be the image of $\Fib L$, which will be called the \emph{primary constraint submanifold}. Let ${\Fib  L}_1$ denote the restriction of $\Fib L$ to $M_1$, that is
    \begin{equation}\label{eq:diagram_FL1}
        \begin{tikzcd}
            TQ \times \RR \arrow[r, "\Fib L"] \arrow[rd, "\Fib L_1"] & T^*Q \times \RR \\
           & M_1 \arrow[u, "g_1"', hook]
            \end{tikzcd}
    \end{equation}
    where $g_1: M_1 \toinj TQ \times \RR$ is the canonical inclusion.

    The submanifold $M_1$ is equipped with the form $\eta_1 = {g_1}^*(\eta_Q)$, where $\eta_Q$ is the canonical contact form in $T^*Q \times \RR$. By the commutativity of the diagram in \cref{eq:diagram_FL1}, we deduce
    \begin{equation}
        {{(\Fib L}_1)}^*(\eta_1)={(\Fib L)}^*(\eta_Q) = \eta_L
    \end{equation}

    \begin{proposition}
        Let $L:P_1 = TQ \times \RR  \to \RR$ be an almost regular Lagrangian such that $\Fib L_1$. Then $\eta_1 = {g_1}^*(\eta_Q)$ is a precontact form of the same class as $\eta_L$.
    \end{proposition}
    \begin{proof}
        Assume $\eta_L$ is of class $2r+1$. Then,  $\eta_1 \wedge \dd \eta_1^r$ is nowhere zero because its image by ${(\Fib L)}^*$ is nowhere zero. 
        
        Also, $\eta_1 \wedge \dd \eta_1^{r+1}$ is everywhere zero. Let $p\in M_1$. Since $\Fib L_1 : P_1 \to M_1$ is a submersion, there are smooth local sections $G:U \to P_1$, where $p \in U \subseteq M_1$ such that $\Fib L_1 \comp G = \Id_U$. Then,
        \begin{equation}
            0 = G^* (\eta_L \wedge \dd {\eta_L}^{r+1}) = 
            G^*({(\Fib L_1)}^*(\eta_1 \wedge \dd \eta_1^{r+1})) =
            \eta_1 \wedge \dd \eta_1^{r+1}.
        \end{equation}
        
        Therefore $\eta_1$ is a precontact form of class $2r + 1$.
    \end{proof}

    The last ingredient for setting up a precontact Hamiltonian system on $M_1$ is a Hamiltonian function $H_1 : M_1 \to \RR$. By requiring that $\Fib L$ has connected fibers we obtain the following result:
    \begin{proposition}\label{thm:ker_FL_energy}
        Let $L: \RR\times TQ \to \RR$, be an almost regular Lagrangian, then, there is a unique function $H_1: M_1 \to \RR$ such that the following diagram commutes:
        \begin{equation}\label{eq:diagram_FL1_energy}
            \begin{tikzcd}
                & \RR \\
                TQ \times \RR \arrow[r, "\Fib L"] \arrow[rd, "\Fib L_1"]
                \arrow[ur,"E_L"] 
                & T^*Q \times \RR \\
               & M_1 \arrow[u, "g_1"', hook] \arrow[uu,"H_1", swap, bend right = 60]
                \end{tikzcd}
        \end{equation}
        That is, 
        \begin{equation}\label{eq:Ham_implicit}
            H_1 \comp \Fib L = E_L, 
        \end{equation}
    \end{proposition}

    \begin{proof}

        We will prove that $E_L$ is constant along the fibers of $\Fib L$, so $H_1$ is well defined. Since the fibers are connected, it is enough to see that $\lieD{Z} E = 0$ for every $Z \in \ker {(\Fib L)}_*$.

        One can compute (see~\cite[page~3424]{Chinea1994})
        \begin{equation}\label{FLkernel}
            \ker{(\Fib L)}_* = 
            \ker\dd \lambda_L \cap \im S = 
            \ker\dd \eta_L \cap \im S = \CharD \cap \im S.
        \end{equation}
        
        In bundle coordinates, one can see that $X \in \ker{(\Fib L)}_*$ if and only if
        \begin{equation}
            X = b^j \pdv{}{\dot{q}^j}, 
        \end{equation}
        where, for all $i$,
        \begin{equation}
            b^j \frac{\partial^2 L}{\partial \dot{q}^i \partial \dot{q}^j} =0.
        \end{equation}

        By using the coordinate expression of the energy (\cref{eq:energy_coords}) we find that
        \begin{equation}
            X(E_L) =
            \dot{q}^i b^j  \frac{\partial^2 L}{\partial \dot{q}^i \partial \dot{q}^j}  = 0.
        \end{equation}
\begin{comment}
            Since $Z\in \im S$, $Z=JW$ locally for some field $W$. By \cref{lieDJfunc}
            \begin{align*}
                \lieD{Z} E &= \lieD {SW}(\bracks{\dd L, \Delta}) - (JW)(L) \\ &=
                \contr{W} (S^* \dd \bracks{\dd L, \Delta} - S^* \dd L), 
            \intertext{by \cref{lieBrackDeltaJ}:}  &=
                \contr{W}(-\contr{\Delta} \lieD{S} \dd L + \contr{S} \dd L - S^* \dd{L}) \\&=
                - \dd \lambda_L(W, \Delta),
            \end{align*}
            Now, since $\Delta \in \im S$, locally $\Delta = S Y$ for some vector field $Y$. Hence, by \cref{vertical_forms,FLkernel}
            \begin{align*}
                - \dd \lambda_L (W,JY) &=  
                \dd \lambda_L(JW,Y) = 
                \dd \lambda_L(Z,Y) = 0.
            \end{align*}
\end{comment}
    \end{proof}
    
    By the results of this chapter we conclude that if the Lagrangian is almost regular, then $(M_1, \eta_1, H_1)$ is a precontact Hamiltonian system. Thus, we apply the constraint algorithm (\cref{sec:constr_alg}) to the equation $\bar{\flat}_1(Y) = \gamma_{H_1}$, where $\bar{\flat}_1$ is the mapping defined by $\eta_1$. Thus we obtain a sequence of constraint submanifolds
    \begin{equation}
        \cdots \toinj M_{i} \toinj \cdots \toinj M_2 \toinj M_1,
    \end{equation}
    where $M_f$ is the final constraint submanifold.

    We will investigate the connection between the algorithm on the precontact systems $(P_1,\eta_L,E_L)$ and $(M_1, \eta_1, H_1)$.

    \begin{lemma}
        The following diagram commutes
        \begin{equation}
            \begin{tikzcd}
                TQ \times \RR \arrow[r, "\Fib L"] \arrow[rd, "\Fib L_1"] & T^*Q \times \RR \\
                P_2 \arrow[rd, "\Fib L_2"] \arrow[u, "j_2", hook] & M_1 \arrow[u, "g_1"', hook] \\
                \vdots \arrow[u, "j_3", hook] & M_2 \arrow[u, "g_2"', hook] \\
                P_f \arrow[rd, "\Fib L_f"] \arrow[u, "j_f", hook] & \vdots \arrow[u, "g_3"', hook] \\
                 & M_f \arrow[u, "g_f"', hook]
                \end{tikzcd}
        \end{equation}
        where $P_i$ and $M_i$ are the $i$-th constraint submanifolds obtained in the constraint algorithm to $P_1 = TQ \times \RR$ and to $M_1$ respectively, and $j_i:P_i\to P_{i-1}$, $g_i:M_i\to M_{i-1}$ are the canonical inclusions. The submersions $\Fib L_i: P_i \to M_i$ are the restrictions of the Legendre transformation $\Fib L$ to the corresponding constraint submanifolds.
    \end{lemma}
    \begin{proof}
        All the claims follow from proving that $\Fib L_i (P_i)=M_i$. Equivalently, we need two show that the constraint defining $P_i$ are precisely the pullback by $\Fib L_i$ of the constraints defining $M_i$. For performing the algorithm, we choose Reeb vector fields which are $\Fib L$-related.

        First, let ${\omega}_L = \dd \eta_L + \eta_L \otimes \eta_L$ and ${\omega}_1 = \dd \eta_1 + \eta_1 \otimes \eta_1$. Since ${(\Fib L)}^* (\eta_1)= \eta_L$, then ${(\Fib L)}^*(\omega_1) = {\omega}_L$. From this and the fact that $FL_*$ is surjective, it easily follows that $FL_*$ maps  $\orth{T P_i}$ onto $\orth{T M_i}$.

        By taking  $\Fib L$-related Reeb vector fields, from a straightforward computation we find that ${(\Fib L)}^* (\gamma_{H_1}) = \gamma(E_L)$. With this, we have that for any $Y \in \orth{TM_i}$ and any $X\in \orth{TP_i}$ such that ${(\Fib L)}_* X = Y$, ${(\Fib L)}^*(\gamma_{H_1}(Y))= \gamma_{E_L}(X)$. Hence $\Fib L_i(P_i)=M_i$ because their constraints are related by the Legendre transformation.
    \end{proof}

    From the commutativity of the diagram, we get the following result.
    \begin{theorem}[Equivalence Theorem]\label{thm:lagrange_hamilton_equiv}
        Let $L:P\times \RR\to \RR$ be an almost regular Lagrangian, let $(P,\eta_L,E_L)$ be the corresponding precontact system, and let  $(M_1,\eta_1,H_1)$ be its Hamiltonian counterpart,. We denote the final constraint submanifolds by $P_f$ and $M_f$, respectively. Then 
        \begin{itemize}
            \item   For every $\Fib L$-projectable solution $X$ of the equations of motion along $P_f$, ${(\Fib L)}_* (X)$ is a solution of Hamilton equations of motion along $M_f$.
            \item  For every solution $Y$ of Hamilton equations of motion along $M_f$, every $X\in \VecFields(TQ \times \RR)$ such that ${(\Fib L)}_*(X)=Y$ solves the equations of motion along $P_f$.
        \end{itemize}
    \end{theorem}

    \section{The constraint algorithm and the Dirac-Jacobi brackets}\label{sec:DJ_brackets}
    The aim of this section is to develop a local version of the constraint algorithm based on the Jacobi bracket of the contact manifold $T^*Q\times \RR$, similar to the Dirac-Bergmann algorithm for the presymplectic case \cite{Dirac2001,Bergmann1955}. This bracket also has some global geometric descriptions~\cite{Leon1995,Ibort1999a}.
    It has been extended to the time-dependent case in~\cite{Chinea1994}.
      
    The bracket formalism will allow us to classify the constraints produced by the algorithm depending of weather they provide dynamical information (\emph{first class}) or not (\emph{second class}). Furthermore, we will define a modified brackets, the \emph{Dirac-Jacobi bracket} which will provide us expressions for the evolution of the observables witch are manifestly independent on the second class constraints.

    As we have explained in \cref{sec:geometry_and_dynamics}, a contact manifold $(M, \eta)$ is a particular case of a Jacobi manifold, with Jacobi structure $(\Lambda, -\Reeb)$ as in \cref{eq:contact_jacobi}. We remind that the Jacobi bracket is given by
    \begin{equation}
        \jacBr{f,g} = \Lambda{(\dd f, \dd g)} - f\Reeb(g) + g \Reeb(f),
    \end{equation}
     for $f,g \in \Cont^\infty(M)$. We recall that this brackets are not Poisson. Instead, they satisfy the following generalized Leibniz rule:
    \begin{equation}\label{eq:mod_derivation}
        \jacBr{fg,h} = f\jacBr{g,h} + g\jacBr{f,h} + f g \Reeb(h),
    \end{equation}
    for arbitrary functions $f,g,h \in \Cont^\infty(M)$.

    The evolution of an observable $f\in \Cont^{\infty}(T^* Q)$ can be written in terms of its bracket with the Hamiltonian $H$,
    \begin{equation} 
        \dot{f} = {X_H}(f) = \jacBr{H, f} - f \Reeb{(H)},
    \end{equation}
    where $H$ is an arbitrary extension of $H_1$.

     In this section we will be working with the Hamiltonian formulation of a system that is given by an almost regular Lagrangian $L:TQ\times \RR \to \RR$ as in \cref{sec:equiv_problem}.  Assume that we obtain a first constraint submanifold $M_1 = \Fib L (TQ) \subseteq T^*Q$, with a Hamiltonian function $H_1: M_1 \to \RR$. We can extend $H_1$ to $T^*Q \times \RR$ as follows:
    \begin{equation}
        H_T = H + u_a \phi^a,
    \end{equation}
    where $H$ is an arbitrary extension of $H_1$, $\phi^a$ are a set of constraints defining $M_1$ and $u_a$ are Lagrange multipliers. Hence, we can compute the evolution of an observable $f$ with respect to $H_T$:
    \begin{equation}\label{eq:evolution_Jacobi}
        \begin{aligned}
            \dot{f} &=  \jacBr{H_T, f} - f\Reeb{(H_T)} \\ &= 
            \jacBr{H, f} 
            + u_a \jacBr{\phi^a, f} 
            -  f \Reeb(H)
            - f u_a \Reeb(\phi^a) \\ & \qquad
            + \phi^a (\jacBr{u_a, f}+ u_a  \Reeb(f) -f \Reeb(u_a)) \\ & =
            \jacBr{H, f}
            -  f \Reeb(H) 
            + u_a (\jacBr{\phi^a, f} 
            -  f \Reeb(\phi^a))
            + \phi^a \Lambda(\dd u_a, \dd f) \\ & =
            (X_H + u_a X_{\phi^a} )(f) + \phi^a \Lambda(\dd u_a, \dd f),
        \end{aligned}
    \end{equation}
    where we have used the generalized Leibniz rule~\cref{eq:mod_derivation}.

    The constraint algorithm can be locally interpreted in terms of this bracket, in a similar fashion to the Dirac algorithm for the symplectic case~\cite{Gotay1978}.
    \begin{remark}
    A local version of the constraint algorithm for constrained on the extended phase $M_1 \subseteq T^*Q \times \RR$ can be given in terms of the Jacobi bracket as follows.

    First, we demand that the primary constraints should be preserved along the evolution of the system. Geometrically, this means that $X_{H_T}$ should be tangent to $M_1$, that is:
        \begin{equation}\label{eq:contrain_preserv}
           (0=  \dot{\phi}^a = X_{H_T}{(\phi^a)} =  \jacBr{H, \phi_a} + u_b \jacBr{\phi^b,\phi^a})\restr{}{M_1},
        \end{equation}
        since $\phi_b = 0$ on $M_1$. We should demand this condition for all linear combinations of the constraints. Some will be satisfied trivially, others will fix the multipliers $u_b$, and the remaining ones will be independent on the multiples $u_b$. The later take the form $f^\alpha_a  \dot{\phi}^a$, where 
        \begin{equation}
            (f^\alpha_a \jacBr{\phi^a,\phi^b} = 0)\restr{}{M_1}.
        \end{equation}
        If we let $\psi^\alpha = f^\alpha_a  \dot{\phi}^a$, then
        \begin{equation}
            (\psi^\alpha = f^\alpha_a  \jacBr{H, \phi_a} )\restr{}{M_1}.
        \end{equation}
        These new constraints define a further submanifold, $M_2$. We can now modify the Hamiltonian by adding the new constraints $H_T' = H_T + v_\alpha \psi^\alpha$ and iterate this procedure until we get not new constraints.
\end{remark}    

    Let $M_f$ be the final constraint submanifold. We say that a function $f\in\Cont^\infty(T^*Q \times \RR)$ is \emph{first class} if $\restr{\jacBr{f,\phi}}{M_f} = 0$. Denote by $\mathcal{F}\subseteq \Cont^\infty(M)$ to the set of first class functions, which is a subalgebra with respect to the Jacobi bracket since, by the Jacobi identity, if $\psi, \chi \in \mathcal{F}$ and $\phi$ is a constraint, then, along $M_f$,
    \begin{equation*}
        \jacBr{\jacBr{\psi,\chi}, \phi} =
        \jacBr{\jacBr{\psi,\phi},\chi} + \jacBr{\psi,\jacBr{\chi,\phi}} = 0.
    \end{equation*}

    The Hamiltonian $H_T$ is an example of a first class function because of the constraint preservation condition given in \cref{eq:contrain_preserv}.

    We say that a function is \emph{second class} if it is not first class.

    We will show that family of independent constraints $\phi^\alpha$ defining $M_f$ (by independent, we mean that their differentials are linearly independent) we can extract a maximal subfamily of second class constraints the matrix of their Jacobi brackets is non-singular. Modifying the rest of them by taking linear combinations, we get second class constraints that still form an independent family.
    
    Consider the matrix $(\bracks{\phi^\alpha,\phi^\beta})_{\alpha,\beta}$. Assume that it has constant rank $k$ in a neighborhood of $M_f$, that is, up to reordering, the first $k$ rows are linearly independent. Denote by $\phi^{a}$ (with latin indices) those functions and $\phi^{\bar{a}}$ (with overlined latin indices) the rest of them. We use greek indices when we want to refer to every constraint. Then the rest of the rows are linear combinations of the first $k$, that is
    \begin{equation}
        \dirBr{\phi^{\bar{a}}, \phi^\beta} = B^{\bar{a}}_a  \dirBr{\phi^a, \phi^\beta}.
    \end{equation}
    Define 
    \begin{equation}
        {\bar{\phi}}^{\bar{a}} =  {\phi}^{\bar{a}} - B^{\bar{a}}_a \phi^{a}.
    \end{equation}
    Using the generalized Leibniz rule (\cref{eq:mod_derivation}) we can check that these new constraints are first class, so $\phi^a, {\bar{\phi}}^{\bar{a}}$ is a basis of the constraints with the desired properties.

    Now let $C^{ab} = \jacBr{\phi^a,\phi^b}$ and let $C_{ab}$ denote the inverse matrix. We define the \emph{Dirac-Jacobi} bracket such that
    \begin{equation}
        \dirBr{f,g} =  \jacBr{f,g}  - \jacBr{f,\phi^a} C_{ab} \jacBr{\phi^b, g}.
    \end{equation}

    \begin{proposition}
        The Dirac-Jacobi bracket has the following properties:
        \begin{enumerate}
            \item It is a Jacobi bracket (\cref{def:jac_bra}) which satisfies the generalized Leibniz rule
            \begin{equation}\label{eq:derivation_dirac}
                \dirBr{fg,h} = f\dirBr{g,h} + g\dirBr{f,h} + f g \Reeb_{DJ}(h),
            \end{equation}
            where
            \begin{equation}
               \Reeb_{DJ} = \Reeb +  C_{ab} 
               \Reeb(\phi^b) (\sharp_{\Lambda}(\dd \phi^a) + \phi^a \Reeb).
            \end{equation}
            \item The second class constraints $\phi^a$ are Casimir functions for the Dirac-Jacobi bracket.
            \item \label{itm:first_class_dirBr} For any first class function $F$,
                  \begin{equation}
                    \begin{aligned}
                        (\dirBr{F,\cdot} &= \jacBr{F,\cdot})\restr{}{M_f}, \\
                        (\Reeb_{DJ}(F)      &= \Reeb(F)        )\restr{}{M_f}.
                    \end{aligned}
                  \end{equation}
            \item The evolution of observables is given by
            \begin{equation}\label{eq:evolution_Dirac}
                \begin{aligned}        
                    (\dot{f} &=      \dirBr{H, f} - f\Reeb_{DJ} {(H)} +
                    \bar{u}_{\bar{a}} (\dirBr{\bar{\phi}^{\bar{a}}, f} - f \Reeb_{DJ} (\bar{\phi}^{\bar{a}})) \\ & = 
                    (X_{H}+\bar{u}_{\bar{a}}X_{\bar{\phi}^{\bar{a}}})(f)
                    ){\restr{}{M_f}},
                \end{aligned}
            \end{equation}
            where $H:T^*Q \times \RR \to \RR$ is an arbitrary extension of the Hamiltonian $H_1$.
            %TODO {Include the term \emph{Gauge symmetry}}
        \end{enumerate}
        We remark that the motion depends on the multipliers of the first class constraints $\bar{u}_{\bar{a}}$, but it is independent on the multipliers of the second class constraints $u_a$.
    \end{proposition}

    \begin{proof}
        It is clear that the brackets are bilinear and antisymmetric. The Jacobi identity follows from a computation as the one performed by Dirac in~\cite{Dirac1950} for the symplectic case. Moreover, the locality of the Dirac-Jacobi bracket follows from the locality of the bracket associated to the natural Jacobi structure of $T^*Q \times \RR$. Therefore, by \cref{thm:jab_bra_characterization}, there is another Jacobi structure $(\Lambda_{DJ}, \Reeb_{DJ})$ on $T^*Q \times \RR$ such that
        \begin{equation}
            \dirBr{f,g} = \Lambda_{DJ} (\dd f, \dd g)  - f\Reeb_{DJ}(g) +g \Reeb_{DJ} (f).
        \end{equation}
        The vector field $\Reeb_{DJ}$ can be computed by taking into account that
        \begin{align*}
            \Reeb_{DJ}(f) &= \dirBr{f,1} \\& =  \jacBr{f,1} - 
            \jacBr{(f,\phi^a)} C_{ab} \jacBr{\phi^b, 1} \\ &= 
            \Reeb(f) -  C_{ab} (\Lambda(\dd f, \dd \phi^a) - f \Reeb(\phi^a) + \phi^a \Reeb(f)) 
            \Reeb(\phi^b) \\ &=  
            (\Reeb +  C_{ab} \Reeb(\phi^b) (\sharp_{\Lambda}(\dd \phi^a) + \phi^a \Reeb))(f) - f C_{ab} \Reeb(\phi^a)\Reeb(\phi^b),
        \end{align*}
        where $-C_{ab} \Reeb(\phi^a)\Reeb(\phi^b) =  \dirBr{1,1} = 0$, by the antisymmetry of the bracket.

        The fact that $\phi^a$ are Casimir functions follows from a straightforward calculation from de definition of the brackets.

        For proving the statement (\ref{itm:first_class_dirBr}), if $F$ is first class it is clear that both brackets coincide along $M_f$ since they differ by multiples of the Jacobi brackets of $F$ with constraints. For the second part, notice that, along ${M_f}$,
        \begin{equation}
            \Reeb{(F)} = \jacBr{1,F} = \dirBr{1,F} = \Reeb_{DJ}(F).
        \end{equation}

        The last claim follows from the combination of the formula for the evolution of observables \cref{eq:evolution_Jacobi} and  Item~(\ref{itm:first_class_dirBr}). Since the second class constraints $\phi^a$ are Casimir functions, their brackets, including $\Reeb_{DJ}(\phi^a)= \dirBr{1,\phi^a}$, will vanish, so the terms with the corresponding multipliers $u^a$ will not affect the evolution of the observable.
    \end{proof}

    \section{The second order problem}\label{sec:second_order}
    Using the theory developed on the previous section, given an almost regular Lagrangian $L:TQ \times \RR \to \RR$ we are able to develop the constraint algorithm on the Lagrangian side, as well as on the Hamiltonian counterpart starting with the image $M_1=\Fib L(\RR \times TQ)$, the precontact form $\eta_1$ and the restricted Hamiltonian $H_1$. The following diagram summarizes the situation: 
    \begin{equation}
        \begin{tikzcd}
            TQ \times \RR \arrow[r, "\Fib L"] \arrow[rd, "\Fib L_1"] & T^*Q \times \RR \\
            P_f \arrow[u, hook] \arrow[dr, "\Fib L_f"] & M_1 \arrow[u, hook] \\
             & M_f \arrow[u, hook]
            \end{tikzcd}
    \end{equation}
    where $P_f$ and $M_f$ are the final constraint submanifolds on the Lagrangian and Hamiltonian sides, which are the maximal submanifolds in which solutions to the equations of motion
    \begin{subequations}
        \begin{align}
            \bar{\flat}_L{(X)} &= \gamma_{E_L},\\
            \bar{\flat} (Y) &= \gamma_{H_1}
        \end{align}
    \end{subequations}
    exist and are tangent to the respective submanifolds. Both submanifolds are connected by the Legendre transformation $\Fib L_f : P_f \to M_f$, which is a surjective submersion.
   
    \begin{remark}
        Notice that in order to get a solution $X$ on the Lagrangian side we can start with a solution $Y$ and use that $\Fib L_f : P_f \to M_f$ is a fibration to construct $X$ such that ${(\Fib L)}_* X = Y$.
    \end{remark}

    As we know,  if the Lagrangian is regular, the Euler-Lagrange are of second order. That is, the solution $X$ is a so-called a \emph{second order differential equation (SODE)} or a \emph{semispray}~\cite{deLeon2011}. This means that, in bundle coordinates $(q^i, \dot{q}^i, z)$, it has the form
    \begin{equation}\label{eq:SODE}
        X =   \dot{q}^i \pdv{}{q^i} + b^i \pdv{}{\dot{q}^i} + c \pdv{}{z}.
    \end{equation}
    This condition can be written in algebraic terms as follows
    \begin{equation}
        S(X)=\Delta.
    \end{equation}
    
    However, this is not the case for singular Lagrangians. We are interested on finding a submanifold $S$ of $P_f$ and a solution $X$ tangent to $S$ that \emph{satisfies the second order condition along} $S$. That is $S(X)_p=\Delta_p$ at every $p \in S$. This is the so-called \emph{second order problem}, which was studied for presymplectic Lagrangian systems in~\cite{Gotay1980a} and in~\cite{Chinea1994} for time dependent Lagrangians. 

    The connection with Herglotz's equations and the related variational problem is apparent from the next result, which parallels \cite[Theorem~3.5.17]{Abraham1978} in the symplectic case.

    \begin{proposition}
        Let $X$ be a vector field on $TQ\times \RR$ that verifies the second order equation condition along a submanifold $S \subseteq TQ \times \RR$, and let $L:TQ \times \RR \to \RR$ be a Lagrangian. Then, along $S$, $X$ solves the equations of motion for $L$ if and only if it solves Hertglotz's equations. 
    \end{proposition} 
    \begin{proof}
        Indeed, if $X$ satisfies the second order equation condition along $S$, then, along $S$
            \begin{equation}
                X =   \dot{q}^i \pdv{}{q^i} + b^i \pdv{}{\dot{q}^i} + c \pdv{}{z}.
            \end{equation}
        
        If it solves the equations of motion, necessarily $\eta_L(X) = -E_L$. Substituting the coordinate expression of $X$, we find out that $\restr{(c=L)}{S}$. Hence, we can perform the same computation than in the regular case (\cref{eq:lagrangian_solution}). That is, along $S$, the coefficients $b^i$ must satisfy the equation
        \begin{equation}
            {b}^i  \frac{\partial}{\partial \dot{q}^i}\parens*{\frac{\partial L}{\partial \dot{q}^j}} 
            + \dot{q}^i  \frac{\partial}{\partial q^i}\parens*{\frac{\partial L}{\partial \dot{q}^j}} 
            + {L  \frac{\partial}{\partial z}\parens*{\frac{\partial L}{\partial \dot{q}^j}}} - \frac{\partial L}{\partial q^j} =
            \frac{\partial L}{\partial \dot{q}^j} \frac{\partial L}{\partial z},
        \end{equation}
        Hence, an integral curve $(q^i, \dot{q}^i,z)$   satisfies Herglotz's equation along $S$:
        \begin{equation}
            \frac{\dd}{\dd t} \parens*{\frac{\partial L}{\partial \dot{q}^i}} - \frac{\partial L}{\partial q^i} =
            \frac{\partial L}{\partial \dot{q}^i} \frac{\partial L}{\partial z}.
        \end{equation}
            
        The converse follows by reversing the computation.
    \end{proof}

    In this section we will be construct a submanifold $S \subseteq P_f$ along which the equations of motion have a unique solution which is a SODE. The first observation is that $\ker {(\Fib L_f)}_* $ is an involutive distribution. Indeed, it is the vertical distribution of the fibration $\Fib L_f: P_f \to M_f$. By the construction of the constrain submanifolds, we can see that for $x\in P_f$ $\ker {(\Fib L)}_* (x) =\ker {(\Fib L_f)}_*(x) \subseteq T_x P_f$.

    Let $X$ be a vector field on $TQ \times \RR$. We define the \emph{deviation} of $X$ as
    \begin{equation}
        X^* = S(X) - \Delta.
    \end{equation}
    We note that $X^*=0$ if and only if $X$ is a second order equation. The next step in the construction of the solution to the second order problem is the following. 
    \begin{lemma}\label{thm:defectCharD}
        If $X$ is a solution of the equations of motion along $P_f$, then $X^*\in \ker {(\Fib L_f)}_*$.
    \end{lemma}
    \begin{proof}
         Assume that $X$ is written in bundle coordinates $(q^i, \dot{q}^i, z)$ on $TQ \times \RR$ by
         \begin{equation} 
            X =   a^i \pdv{}{q^i} + b^i \pdv{}{\dot{q}^i} + c \pdv{}{z}.
        \end{equation}
        then
        \begin{equation}
            X^* = S(X) - \Delta = (a^i - \dot{q}^i)\pdv{}{\dot{q}^i}.
        \end{equation}

        If we contract both sides of the equation of motion $\bar{\flat}_L{(X)} = \gamma_H$ by $\pdv{}{\dot{q}^i}$ we get
        \begin{equation}
            (a^i-\dot{q}^i) \frac{\partial^2 L}{\partial \dot{q}^i \partial \dot{q}^j} = 0.
        \end{equation}
        Next, we compute
        \begin{equation}\label{eq:defectCharD_proof}
            \Fib L_*(X^*)= (a^j-\dot{q}^j) \frac{\partial^2 L}{\partial \dot{q}^i \partial \dot{q}^j} = 0.
        \end{equation} 

        Hence $X^*$ is tangent to the leaves of the fibration determined by $\ker {(\Fib L_f)}_*$, or, in other words to the fibers of the fibration $\Fib L_f :P_f \to M_f$.
    \end{proof}

    Next we will construct the submanifold $S$.
    Fix a point $y \in M_f$ and let $x$ be an arbitrary point on the leaf over $y$, say $\Fib L_f (x) = y$. Assume that $x=(q^i_0,\dot{q}^i_0,z_0)$ in bundle coordinates.

    Notice that $X$ is projectable, hence along a leaf it can only vary from point to point in a direction tangent to $\ker (\Fib L)_*$. Since  $\ker  (\Fib L)_* \subseteq \im S = \gen{\set{\pdv{}{\dot{q}^i}}_i}$, this implies that $a^i$ and $c$ are constant functions along the leafs and only $b^i$ might change.
    
        Consider the vector field
        \begin{equation}
            -X^* = (\dot{q}^i-a^i)\pdv{}{\dot{q}^i},
        \end{equation}
        and compute the integral curve of $-X^*$ passing through $x$, say
        \begin{equation}
            \sigma(t) = (q^i(t), \dot{q}^i(t),z(t)).
        \end{equation}
        Therefore
        \begin{equation}
            \sigma(0) = (q^i(0), \dot{q}^i(0),z(0))=
            (q^i_0, \dot{q}^i_0,z_0).
        \end{equation}
        This integral curve has to satisfy the system of differential equations:
        \begin{equation}
            \dv{\dot{q}^i}{t} = \dot{q}^i(t)- a^i.
        \end{equation}
        Consequently, the solution passing through $x$ is just
        \begin{equation}
            \sigma(t) = (q^i_0, a^i +\exp(t)(q_0^i-a^i),z_0),
        \end{equation}
        which is entirely contained on the fiber over $y$. In addition, the limit point as $t \to -\infty$,
        \begin{equation}
            \tilde{x} = \lim_{t \to -\infty} \sigma(t),
        \end{equation}
        is also on the same fiber, since fibers are closed. A direct computation shows that
        \begin{equation}
            \tilde{x} = (q_0^i,a^i,z_0),
        \end{equation}
        and that
        \begin{equation}
            S(X)_{\tilde{x}} = \Delta_{\tilde{x}}.
        \end{equation}
      Summarizing, we have constructed a smooth section $\alpha:M_f \to P_f$ of the fibration $\Fib L_f: P_f \to M_f$, by taking $\alpha(y) = \tilde{x}$ for some $x$ on the fiber over $y$ (notice that $\tilde{x}$ only depends on $\Fib L(x)$). By taking $S=\alpha(M_f)$ we have the following.
    
    \begin{theorem}[Second order differential equation]
        Let $L: TQ \times \RR \to \RR$ be an almost regular Lagrangian and let $P_f$ be the final constraint embedded submanifold. Then, there exists a submanifold $S\subseteq P_f$ such that the equations of motion have a unique solution $X\in \VecFields(TM \times \RR)$ satisfying the SODE condition. That is, along $S$,
        \begin{equation}
            \bar{\flat}_L (X)= \gamma_{E_L}, \quad S(X)= \Delta.
        \end{equation}
    \end{theorem}
    
    \section{Examples}\label{sec:examples}
    %\todo[inline]{DO second order}

    \subsection*{Example 1: Cawley's Lagrangian}\label{ssec:ex1}
    The Lagrangian considered by Cawley~\cite{Cawley1979} can be modified by adding a linear dissipative term $\gamma z$, where $\gamma$ is a real number. Let $Q= \RR^3$, $P_1 = TQ \times \RR $ and consider the Lagrangian function $L: P_1 \to \RR$ such that
    \begin{equation}
        L(q^1, q^2, q^3,\dot{q}^1, \dot{q}^2, \dot{q}^3, z) =  \frac{m}{2}  {(\dot{q}^{1} + \dot{q}^{2})}^{2} + \frac{\mu}{2} {(\dot{q}^3)}^2 + V(q^1, q^2, q^3)+ \gamma z,
    \end{equation}
    for some potential function $V$ and some real nonzero constants $m,\mu$. This Lagrangian induces the following precontact structure on $P_1$
    \begin{equation}
        \eta_L = \dd z - m(\dot{q}^{1} + \dot{q}^{2} ) (\dd q^{1} + \dd q^{2}) - \mu \dot{q}^3 \dd q^3
    \end{equation}
    \begin{equation}
        \dd \eta_L = m(\dd q^{1} + \dd q^{2}) \wedge (\dd \dot{q}^{1} + \dd \dot{q}^{2}) + \mu (\dd q^3 \wedge \dd \dot{q}^3),
    \end{equation}
    One can check that $\eta_L \wedge (\dd \eta_L)^2$ is nowhere zero and $\eta_L \wedge (\dd \eta_L)^3=0$, hence $(P_1,\eta_L)$ is a precontact manifold of class $5$, with the corresponding energy function $E_L = \Delta(L) -L$ given by 
    \begin{equation}
        E_L =  \frac{m}{2}  {(\dot{q}^{1} + \dot{q}^{2})}^{2} + \frac{\mu}{2} {(\dot{q}^3)}^2 - V(q^1, q^2, q^3) - \gamma z.
    \end{equation}
    We now apply the constraint algorithm to the precontact Hamiltonian system $(P_1,\eta_L,E_L)$, choosing the following Reeb vector field:
    \begin{equation}
        \Reeb = \pdv{}{z}.
    \end{equation}

    In order to compute the constraints, we find the complement of the tangent bundle of $P_1$,
    \begin{equation}
        \orth{TP_1} = \ker \eta_L \cap \ker \dd \eta_L= \gen*{\pdv{}{q^1}-\pdv{}{q^2}, \pdv{}{\dot{q}^1}-\pdv{}{\dot{q}^2}}.
    \end{equation}

    By imposing $\gamma_{E_L}(X) = 0$ for $X \in \orth{TP_1}$, we get the following constraint,
    \begin{equation}
        \phi^1 = - \pdv{V}{q^1} + \pdv{V}{q^2},
    \end{equation}
    which defines the submanifold $P_2 = \set{p \in M \mid \phi^1(p) = 0}$. Its tangent space is given by
    \begin{equation}
        TP_2 = \gen*{\pdv{}{\dot{q}^1}, \pdv{}{\dot{q}^2}, \pdv{}{\dot{q}^3}, \pdv{}{z}, \pdv{\phi^1}{q^2} \pdv{}{q^1} - \pdv{\phi^1}{q^1} \pdv{}{q^2}, \pdv{\phi^1}{q^3} \pdv{}{q^1} - \pdv{\phi^1}{q^1} \pdv{}{q^3}}.
    \end{equation}

    The complement is given by
    \begin{equation}
        \orth{TP_2} = \gen*{\pdv{}{q^1}-\pdv{}{q^2}, \pdv{}{\dot{q}^1}-\pdv{}{\dot{q}^2}}.
    \end{equation}
    Demanding  $\gamma_{E_L}(X) = 0$ for $X \in \orth{TP_2}$ produces no new constraints, hence $P_2=P_3 = P_f$ and the algorithm ends.

    Notice that the $\Reeb$ vector field is already tangent to the submanifold, so we would ge the same result by using the modified version of the algorithm with imposes the tangency of $\Reeb$.

    \subsubsection*{Hamiltonian formulation and the Legendre transformation}
    For this Lagrangian system, the Legendre transformation is given by $\Fib L: TQ \times \RR \to T^*Q \times \RR$,
    \begin{equation}
        \Fib L(q^1, q^2, q^3,\dot{q}^1, \dot{q}^2, \dot{q}^3, z) = 
        (q^1, q^2, q^3, m(\dot{q}^1 + \dot{q}^2), m(\dot{q}^1 + \dot{q}^2), \mu \dot{q}^3, z)
    \end{equation}
    We obtain that 
    \begin{equation}
        \ker {(\Fib L)}_* = \bracks*{\pdv{}{\dot{q}^1}-\pdv{}{\dot{q}^2}},
    \end{equation} 
    hence it is a submersion onto its image and its fibers are connected, so the Lagrangian system is almost regular. By the Equivalence theorem (\cref{thm:lagrange_hamilton_equiv}), there is a Hamiltonian formulation of the problem. The first constraint submanifold is given by $M_1 = \Fib L(P)$ and can be described by the following constraint function
    \begin{equation}
        \psi^1 = p_1 - p_2.
    \end{equation}
    The unique Hamiltonian function $H_1: M_1 \to \RR$ such that $H_1 \comp \Fib L = E_L$ is given by
    \begin{equation}
        H_1 = \frac{1}{2m} {p_1}^2 + \frac{1}{2\mu} {p_3}^2 - V(q_1,q_2,q_3) - \gamma z.
    \end{equation}

    Let $\eta_1 = g_1^* (\eta)$, where $g_1: M_1 \toinj T^*Q \times \RR$ is the inclusion and $\eta$ is the canonical contact form on $T^*Q \times \RR$, then $(M_1, \eta_1, H_1)$ is a precontact Hamiltonian system. We can apply the algorithm to compute the secondary constraints or use the commutativity of the diagram on \cref{thm:lagrange_hamilton_equiv}. We now obtain a secondary constraint submanifold given by $M_2 = \set{p \in M_1 \mid \psi^2(p) = 0}$, where
    \begin{equation}
        \psi^2 = \pdv{V}{q^1} - \pdv{V}{q^2}.
    \end{equation}

    The algorithm now ends, as $M_3 = M_2=M_f$.

    \subsubsection*{The Dirac-Jacobi bracket}
    We compute the bracket
    \begin{equation}
        \jacBr{\psi^1,\psi^2} = \frac{\partial^2\,V}{\partial {(q^{1})} ^ 2} - 2 \, \frac{\partial^2\,V}{\partial q^{1}\partial q^{2}} + \frac{\partial^2\,V}{\partial {(q^{2})}^ 2},
    \end{equation}
    which is the determinant of the Hessian matrix of $V$ with respect to $(q^1,q^2)$. We will assume that this bracket does not vanish along $M_f$, hence both constraints are second class.

    We will call $F =  \jacBr{\psi^1,\psi^2}$. The Dirac-Jacobi bracket is given by
    \begin{equation}
        \dirBr{f,g} = \jacBr{f,g} + \frac{
              \jacBr{f,\psi^1}\jacBr{\psi^2,g} - \jacBr{f,\psi^2}\jacBr{\psi^1,g}}{F}
    \end{equation} 
    
    The non-zero Dirac-Jacobi brackets of the coordinate functions are
\begin{subequations}
        \begin{align}
            \dirBr{q^1,p_1} &=  \dirBr{q^1,p_2} = \frac{\pdv{\psi^2}{q^2}}{F}\\
            \dirBr{q^2,p_1} &=  \dirBr{q^2,p_2} = -\frac{\pdv{\psi^2}{q^1}}{F}\\
            \dirBr{q^1,p_3} &= -\dirBr{q^2,p_3} = \frac{\pdv{\psi^2}{q^3}}{F} \\
            \dirBr{q^3,p_3} &= -1 \\
            \dirBr{q^1,z} &= -q^1 + \frac{\psi^2}{F} = -q^1 \; \text{along $M_f$}\\
            \dirBr{q^2,z} &= -q^2 - \frac{\psi^2}{F} = -q^2 \; \text{along $M_f$}\\
            \dirBr{q^3,z} &= -q^3.
        \end{align}
\end{subequations}
    With those brackets, we can easily compute the equations of motion along the constrained submanifold $M_f$,
    
    \begin{subequations}\label{eq:sols1}
        \begin{align}
            \dot{q}^1 &=  
                      - \frac{ \frac{p_1}{m} \pdv{\psi^2}{q^2} + \frac{p_3}{\mu} \pdv{\psi^2}{q^2}}{F} \\
            \dot{q}^2 &= \frac{p_1}{m} + 
            \frac{ \frac{p_1}{m} \pdv{\psi^2}{q^2} + \frac{p_3}{\mu} \pdv{\psi^2}{q^2}}{F} \\
            \dot{q}^3 &=  \frac{p_3}{\mu}\\
            \dot{p}_i &= \pdv{V}{q^i} + \gamma p_i \\
            \dot{z} &= -\frac{1}{2m} {p_1}^2 - \frac{1}{2\mu} {p_3}^2 - V(q_1,q_2,q_3) + \gamma z
        \end{align}
    \end{subequations}

    \subsubsection*{The second order problem}
    Consider the vector field $Y$ associated to the equations \cref{eq:sols1}. That is,
    \begin{equation}
            Y =  \dot{q}^i \pdv{}{q^i} + \dot{p}^i \pdv{}{p^i} + \dot{z} \pdv{}{z},
    \end{equation}
    where $(\dot{q}^i,\dot{p}^i,\dot{x})$ are those from \cref{eq:sols1}.
    The vector field $X$ is a solution to the equations of motion on $TQ\times \RR$ that satisfies $(\Fib L)_* X = Y$ and is given by
    \begin{equation}
        \begin{aligned}
            X =& - \frac{ 2\dot{q}^1 \pdv{\psi^2}{q^2} + \dot{q}^3 \pdv{\psi^2}{q^2}}{F} \pdv{}{q^1} \\
            &+ \parens*{2 \dot{q}^1 + \frac{ 2\dot{q}^1 \pdv{\psi^2}{q^2} + \dot{q}^2 \pdv{\psi^2}{q^3}}{F}} \pdv{}{q^2} \\ &+
            \dot{q}^3 \pdv{}{q^3} \\ &+
            \parens*{\pdv{V}{q^1} + 2 m \gamma \dot{q}^1} \pdv{}{\dot{q}^1}\\ &+
            \parens*{\pdv{V}{q^3} +  \mu \gamma \dot{q}^3} \pdv{}{\dot{q}^3} \\ &+
            \parens*{2 m {(\dot{q}^1)}^2 + \mu{(\dot{q}^1)}^2 - V + \gamma z} \pdv{}{z}. 
        \end{aligned}
    \end{equation}

    We will construct the section $\alpha$ of $\Fib L_f$. Notice that, by  the first constraint $p_1=p_2$ on $M_f$, hence any point on $M_f$ has the form $y=(q^1,q^2,q^3,p_1,p_1,p_3,z)$. Take $x = (q^1,q^2, q^3, {p_1}/{m},0, p_3/\mu,z)\in P_f$ so that $\Fib L(x)=y$. We set $\alpha(y)=\tilde{x}$, that is,
    \begin{equation}
        \begin{aligned}
            \alpha&(q^1,q^2,q^3,p^1,p^2,p^3,z) = \biggl( q^1,q^2,q^3, \\ 
            &- \frac{ 2\frac{p_1}{m} \pdv{\psi^2}{q^2} + \frac{p_3}{\mu} \pdv{\psi^2}{q^2}}{F},
            2 \frac{p_1}{m} + \frac{ 2\frac{p_1}{m}^1 \pdv{\psi^2}{q^2} + \frac{p_3}{\mu} \pdv{\psi^2}{q^2}}{F},
            \frac{p_u}{\mu}, z \biggr).
        \end{aligned}
    \end{equation}
    Hence $X$ is satisfies the SODE condition along $\im \alpha$.

    \subsection*{Example 2}\label{ssec:ex_no_reeb}
    Let $Q= \RR^2$, $P_1 = TQ \times \RR $ and consider the Lagrangian function $L: P_1 \to \RR$ defined by
    \begin{equation}
        L(q^1.q^2,\dot{q}^1, \dot{q}^2, z) =  \frac{1}{2}  {(\dot{q}^{1} + \dot{q}^{2})}^{2} + q^{1} + q^{2} z.
    \end{equation}
    This Lagrangian induces the following precontact structure of class $3$
    \begin{equation}
        \eta_L = \dd z - (\dot{q}^{1} + \dot{q}^{2} ) (\dd q^{1} + \dd q^{2}) 
    \end{equation}
    \begin{equation}
        \dd \eta_L = (\dd q^{1} + \dd q^{2}) \wedge (\dd \dot{q}^{1} + \dd \dot{q}^{2}) 
    \end{equation}
    \begin{equation}
        E_L = \frac{1}{2}  {(\dot{q}^{1} + \dot{q}^{2})}^{2} - q^{1} - q^{2} z
    \end{equation}

    We choose the following Reeb vector field,
    \begin{equation}
        \Reeb = \pdv{}{z}.
    \end{equation}

    As in the previous example, we apply the algorithm, obtaining the following constrained submanifolds. Since
    \begin{equation}
        \orth{TP_1} = \gen*{\pdv{}{q^1}-\pdv{}{q^2}, \pdv{}{\dot{q}^1}-\pdv{}{\dot{q}^2}},
    \end{equation}
    then $P_2 = \set{p \in M \mid \phi^1(p) = 0}$, where
    \begin{equation}
        \phi^1 = z - 1.
    \end{equation}

    \begin{equation}
        TP_2 = \gen*{\pdv{}{q^1}, \pdv{}{q^2}, \pdv{}{\dot{q}^2}, \pdv{}{\dot{q}^2}}
    \end{equation}

    \begin{equation}
        \orth{TP_2} = \gen*{\pdv{}{q^1}-\pdv{}{q^2}, \pdv{}{\dot{q}^1}-\pdv{}{\dot{q}^2}, (\dot{q}^1+\dot{q}^2)\pdv{}{\dot{q}^1} + \pdv{}{z}}
    \end{equation}

    $P_3 =  \set{p \in M \mid \phi^2(p) = \phi^2(p) = 0}$,
    \begin{equation}
        \phi_2 = L - 2q^2 =  \frac{1}{2}  {(\dot{q}^{1} + \dot{q}^{2})}^{2} + q^{1} + q^{2} (z-2)
    \end{equation}

    \begin{equation}
        TP_3 = \gen*{(2-z)\pdv{}{q^1} + \pdv{}{q^2}, 
        \pdv{}{\dot{q}^1} - \pdv{}{q^2},
        (\dot{q}^1 + \dot{q}^2) \pdv{}{q^1} - \pdv{}{\dot{q}^1}},
    \end{equation}

    \begin{equation}
        \orth{TP_3} = \gen*{\pdv{}{q^1}-\pdv{}{q^2}, \pdv{}{\dot{q}^1}-\pdv{}{\dot{q}^2}, (\dot{q}^1+\dot{q}^2)\pdv{}{\dot{q}^1} + \pdv{}{z}},
    \end{equation}
    so we get no new constraints and the algorithm ends.

    We remark that any Reeb vector field $\Reeb$ satisfies $\Reeb(\phi^1)=1$, hence if we imposed the tangency of $\Reeb$, we would get the empty set.

    \subsubsection*{Hamiltonian formulation and the Legendre transformation}
    The Legendre transformation is given by
    \begin{equation}
        \Fib L(q^1, q^2, \dot{q}^1, \dot{q}^2, z) = 
        (q^1, q^2, \dot{q}^1 + \dot{q}^2, \dot{q}^1 + \dot{q}^2, z),
    \end{equation}
    and then
    \begin{equation}
        \ker {(\Fib L)}_* = \bracks*{\pdv{}{\dot{q}^1}-\pdv{}{\dot{q}^2}}.
    \end{equation} 
    In this case, the first constraint submanifold $M_1 = \Fib L(P)$ is described by the constraint
    \begin{equation}
        \psi^1 = p_1 - p_2.
    \end{equation}

    The corresponding Hamiltonian $H_1: M_1 \to \RR$ is given by 
    \begin{equation}
        H_1 =  \frac{1}{2}  (p_1)^{2} - q^{1} - q^{2} z.
    \end{equation}

    By the correspondence with the Lagrangian formulation, there will be two constraint submanifolds, $M_3 \toinj M_2 \toinj M_1$, defined by the constraint functions
    \begin{align}
        \psi^2 &= z-1, \\
        \psi^3 &= \frac{1}{2}  {(p_1)}^{2} + q^{1} + q^{2} (z-2),
    \end{align}
    respectively.

    \subsubsection*{The Dirac-Jacobi bracket}
    The constraints have the following Dirac brackets.
    \begin{subequations}
        \begin{align}
            \dirBr{\psi^1,\psi^2} &= 0 \\
            \dirBr{\psi^1,\psi^3} &= 3 -z = 2 \; \text{along $M_f$} \\
            \dirBr{\psi^1,\psi^3} &= -\frac{1}{2} {(p^1})^{2} + q^{1} - q^{2} =
                                     - 2 (q^{2} -q^{1}) \; \text{along $M_f$}
        \end{align}
    \end{subequations}
    \begin{comment}
        \begin{equation}
            \begin{aligned}
                (\dirBr{\psi^\alpha,\psi^\beta})_{\alpha,\beta} & =            
            \left(\begin{array}{ccc}
                0 & 0 & z - 3 \\
                0 & 0 & \frac{1}{2} \, {(p^1)}^{2} - q^{1} + q^{2} \\
                -z + 3 & -\frac{1}{2} \, {(p^1})^{2} + q^{1} - q^{2} & 0
            \end{array}\right) \\ &=
            \left(\begin{array}{ccc}
                0 & 0 & - 2 \\
                0 & 0 &  2 (q^{2} -q{1}) \\
                2 &  - 2 (q^{2} -q^{1}) & 0
            \end{array}\right)
            \, \text{along $M_f$.}
        \end{aligned}
        \end{equation},    
\end{comment}

    The rank of the matrix $(\dirBr{\psi^\alpha,\psi^\beta})_{\alpha,\beta}$ is $2$, so we can extract one first class constraint as a $\Cont^\infty$-linear combination. We set
    \begin{equation}
        \bar{\chi} = \psi^2 - (q^2 - q^1)\psi^0 =  (p^1-p^2)(q^2-q^1)+z-1,
    \end{equation}
    which is a first class constraint, and the other two are second class, which we will relabel $\chi^1 = \psi^1$, $\chi^2 = \psi^3$.

    The Dirac-Jacobi bracket is given by
    \begin{equation}
        \begin{aligned}
            \dirBr{f,g} &=  \jacBr{f,g} + \frac{
                  \jacBr{f,\psi^1}\jacBr{\psi^2,g} - \jacBr{f,\psi^2}\jacBr{\psi^1,g}}{-\frac{1}{2} {(p^1})^{2} + q^{1} - q^{2}} \\ &= 
                  \jacBr{f,g} - \frac{\jacBr{f,\psi^1}\jacBr{\psi^2,g} - \jacBr{f,\psi^2}\jacBr{\psi^1,g}}{2 (q^{2} -q^{1})}
                  \quad \text{along $M_f$.}
        \end{aligned}
    \end{equation}
    Notice that the denominators do not vanish along the submanifold. 

    The non-zero Dirac-Jacobi brackets of the coordinate functions, along $M_f$ are the following
    \begin{subequations}
        \begin{align}
            \dirBr{q^1,q_2} &= -\frac{1}{2}(q^1q^2 + {(q^2)}^2 +p_1)\\
            \dirBr{q^1,p_1} &=  \dirBr{q^1,p_2} = \frac{1}{2}\\
            \dirBr{q^2,p_1} &=  \dirBr{q^2,p_2} = -\frac{1}{2}\\
             \dirBr{q^1,z} &= -\frac{3}{2} q^2\\
            \dirBr{q^2,z} &= -q^1 + \frac{1}{2}q^2.
        \end{align}
    \end{subequations}
        \begin{comment}
            \begin{gathered}
               \mbox{\scriptsize%            
            $\begin{pmatrix}
                0       & 2(p^2-p^1) & 
                -\frac{1}{2} {(p^1)}^{2} + p^1 p^2 + (q^{1} - q^{2}) (2-z) \\
                -2(p^2-p^1) & 0       &  z-3 \\
                \frac{1}{2} {(p^1)}^{2} - p^1 p^2 - (q^{1} - q^{2}) (2-z)& -z + 3 & 0
            \end{pmatrix}$} \\ =
            \left(\begin{array}{ccc}
                0 & 0 &  0 \\
                0 & 0 &  2  \\
                0 &  - 2 & 0
            \end{array}\right)
            \, \text{along $M_f$.}
        \end{gathered}
        \end{equation}
        \end{comment}
        Now consider the total Hamiltonian $H_T: T^*Q \times \RR \to \RR$
        \begin{equation}
            H_T = \frac{1}{2}  (p_1)^{2} - q^{1} - q^{2} z + \bar{u} \bar{\chi} = H_0 + \bar{u} \bar{\chi},
        \end{equation}
        where $\bar{u}$ is an unspecified Lagrange multiplier and 
        \begin{equation}
            H_0 = \frac{1}{2}  (p_1)^{2} - q^{1} - q^{2} z.
        \end{equation}
        The other two Lagrange multipliers are irrelevant for the motion, so we can eliminate them. We can compute the equations of motion. By~\cref{eq:evolution_Dirac}, for any observable $f$ along $M_f$
        \begin{equation}
            \begin{aligned}
                \dot{f} &= \dirBr{H_0, f} -f \Reeb_{DJ}{(H_0)}  
                     + (\dirBr{\bar{\chi}, f} - f\Reeb_{DJ}(\bar{\chi})) \bar{u}
                      \\ &=
                   \dirBr{H_0, f} - q_2 f  +
                   (\dirBr{\bar{\chi}, f} -  f) \bar{u} .
            \end{aligned}
        \end{equation}
        Below we compute the equations of motion along $M_f$,
        \begin{subequations}\label{eq:sols2}
            \begin{align}
                \dot{q}^1 &= (q^1-q^2)  q^1 q^2 - 2{(q^2)}^2 
                    + (q_2-q_1)\bar{u} \\
                \dot{q}^2 &= \frac{1}{2} {(p_1)}^2 q^2 + {(q^2)}^2 + p_1 
                    + (q_1-q_2)\bar{u}\\
                \dot{p}_1 &= \dot{p}_2 = p_1 q^2 +1 -  (p_1)\bar{u}\\
                \dot{z} &= 0.
            \end{align}
        \end{subequations}
        \subsubsection*{The second order problem}
        Consider the vector field $Y$ associated to the equations \cref{eq:sols2} with $\bar{u}=0$. That is,
        \begin{equation}
            \begin{aligned}
                Y =& (q^1-q^2)  q^1 q^2 - 2{(q^2)}^2 \pdv{}{q^1}\\
                &+ \parens*{\frac{1}{2} {(p_1)}^2 q^2 + {(q^2)}^2 + p_1} \pdv{}{q^2} \\
                &+ (p_1 q^2 +1)\parens*{\pdv{}{p_1}+\pdv{}{p_2}}.
            \end{aligned}
        \end{equation}
        The vector field $X$ is a solution to the equations of motion on $TQ\times \RR$ that satisfies $(\Fib L)_* X = Y$ and is given by
        \begin{equation}
            \begin{aligned}
                X =& \parens*{(q^1-q^2)  q^1 q^2 - 2{(q^2)}^2} \pdv{}{q^1}\\
                &+ \parens*{{2} {(\dot{q}^1)}^2 q^2 + {(q^2)}^2 + 2\dot{q}^1} \pdv{}{q^2} \\
                &+ (2\dot{q}^1 q^2 +1)\pdv{}{\dot{q}^1}.
            \end{aligned}
        \end{equation}
        We will construct the section $\alpha$ of $\Fib L_f$. Notice that, by  the first constraint $p_1=p_2$ on $M_f$, hence any point on $M_f$ has the form $y=(q_1,q_2,p_1,p_1,z)$. Take $x = (q_1,q_2,p_1,0,z)\in P_f$ so that $\Fib(x)=y$. We set $\alpha(y)=\tilde{x}$, that is,
        \begin{equation}
            \alpha(q^1,q^2,p^1,p^2,z) = (q^1,q^2,
            (q^1-q^2)  q^1 q^2 - 2{(q^2)}^2,
            {2} {(\dot{q}^1)}^2 q^2 + {(q^2)}^2 + 2\dot{q}^1,
            z).
        \end{equation}
        Hence $X$ satisfies the SODE condition along $\im \alpha$.
\section*{Acknowledgments}
This  work  has  been  partially  supported  by the MINECO  Grants  MTM2016-76-072-P and the ICMAT Severo Ochoa projects SEV-2011-0087 and SEV-2015-0554.  Manuel Lainz wishes to thank ICMAT and UAM for a FPI-UAM predoctoral contract.

\printbibliography

% \makeatletter
% \providecommand\@dotsep{5}
% \makeatother
% \listoftodos\relax
\end{document}